\documentclass[journal,10pt,twocolumn,twoside]{IEEEtran}
\usepackage{amsmath,amsfonts,amssymb,amsbsy,bm,paralist,theorem,color}
\usepackage{graphicx,algorithmic,algorithm,relsize}
\usepackage{multicol}
\usepackage{multirow}
\usepackage[caption=false,font=normalsize,labelfont=sf,textfont=sf]{subfig}
\usepackage{cite}
\usepackage{hyperref}
\usepackage{cuted}
\usepackage{pifont}
\usepackage{comment}
\usepackage{xcolor}
\usepackage{flushend}

\newtheorem{rem}{Remark}

\newtheorem{prop}{Proposition}

\newcommand{\di}{\text{diag}}

\DeclareMathOperator*{\Tr}{Tr}
\graphicspath{{fig/}}

\definecolor{orange}{RGB}{255,107,0}
\definecolor{green}{RGB}{0,160,20}

\begin{document}
\title{Pinching Antennas in Blockage-Aware Environments: Modeling, Design, and Optimization}

\author{Ximing Xie,~\IEEEmembership{Member,~IEEE}, Fang Fang,~\IEEEmembership{Senior Member,~IEEE}, Zhiguo Ding,~\IEEEmembership{Fellow,~IEEE}, \\ and Xianbin Wang,~\IEEEmembership{Fellow,~IEEE}

\thanks{Ximing Xie, Fang Fang and Xianbin Wang are with the Department of Electrical and Computer Engineering, and Fang Fang is also with the Department of Computer Science, Western University, London, ON N6A 3K7, Canada (e-mail: \{xxie269, fang.fang, xianbin.wang\}@uwo.ca).}
\thanks{Zhiguo Ding is with the University of Manchester, Manchester, M1 9BB, UK, and Khalifa University, Abu Dhabi, UAE. (e-mail:zhiguo.ding@ieee.org).}

}\maketitle

\begin{abstract}
Pinching-antenna (PA) systems have recently emerged as a promising member of the flexible-antenna family due to their ability to dynamically establish line-of-sight (LoS) links. While most existing studies assume ideal environments without obstacles, practical indoor deployments are often obstacle-rich, where LoS blockage significantly degrades performance. This paper investigates pinching-antenna systems in blockage-aware environments by developing a deterministic model for cylinder-shaped obstacles that precisely characterizes LoS conditions without relying on stochastic approximations. Based on this model, a special case is first studied where each PA serves a single user and can only be deployed at discrete positions along the waveguide. In this case, the waveguide–user assignment is obtained via the Hungarian algorithm, and PA positions are refined using a surrogate-assisted block-coordinate search. Then, a general case is considered where each PA serves all users and can be continuously placed along the waveguide. In this case, beamforming and PA positions are jointly optimized by a weighted minimum mean square error integrated deep deterministic policy gradient (WMMSE-DDPG) approach to address non-smooth LoS transitions. Simulation results demonstrate that the proposed algorithms significantly improve system throughput and LoS connectivity compared with benchmark methods. Moreover, the results reveal that pinching-antenna systems can effectively leverage obstacles to suppress co-channel interference, converting potential blockages into performance gains.

\end{abstract}


\section{Introduction}
As wireless networks evolve toward the next generation, numerous emerging technologies have been proposed to meet the stringent requirements of high throughput, high reliability, and low latency \cite{10054381}. Among these, flexible-antenna technologies such as reconfigurable intelligent surfaces (RISs) \cite{wuRIS}, fluid antennas \cite{wongfluid}, and movable antennas \cite{zhumoveable} have attracted particular attention due to their ability to dynamically reconfigure wireless channels. In particular, RISs employ multiple passive elements to reflect incident signals and reshape the wireless channel by adjusting their phase shifts \cite{xieRIS}. In contrast, fluid antennas dynamically alter their physical shape \cite{wongfluid2}, and movable antennas adjust their position \cite{zhumove2}, both enabling direct reconfiguration of the wireless channel. These flexible-antenna systems offer performance advantages in data rate, coverage, and outage over conventional fixed-antenna systems and are therefore considered promising candidates for deployment in 6G networks. \par

Although flexible-antenna systems can enhance wireless performance, they still face significant limitations in mitigating large-scale path loss, especially when the line-of-sight (LoS) link is blocked \cite{dingpinchi}. Specifically, RISs create a virtual LoS link between the transmitter and receiver through reflective elements; however, their effectiveness is often constrained by severe attenuation resulting from long propagation distances. In contrast, fluid and movable antennas can only adjust their positions within a few wavelengths, which limits their ability to establish reliable LoS links. To address this challenge, pinching antennas (PAs) have recently been proposed as a promising member of the flexible-antenna family by DOCOMO in 2022 \cite{suzuki2022pinching}. A pinching-antenna system typically comprises one or more dielectric waveguides, along which multiple low-cost plastic pinches are deployed as additional dielectric elements. Electromagnetic waves are radiated from the locations where these pinches are placed, thereby enabling the establishment of LoS links through the dynamic positioning of pinches at desired locations \cite{liu2025pinching}. As a result, the PA placement design in pinching-antenna systems is an important research direction. \par

Beyond their capability of establishing LoS links, pinching-antenna systems have also been shown to exploit obstacles for interference management \cite{dingblockage}. In particular, a PA can be positioned to maintain a LoS link with its intended user, while obstacles simultaneously block LoS paths to other users, thereby mitigating co-channel interference. This ability makes PA systems different from other flexible-antenna technologies and provides new possibilities for interference management, particularly in indoor environments where obstacles are prevalent. Typical examples include shopping malls, airports, office buildings, factories and warehouses, where numerous pillars, walls and racks frequently block the LoS between conventional base station antennas and users. In such environments, a dielectric waveguide integrated into ceilings or walls can host multiple movable pinching antennas and create radiation points close to user hotspots with reduced blockage. However, to realize these gains, PA placement must be carefully designed, since performance depends on user locations, obstacle geometry, and channel conditions. This requires joint consideration of PA placement, waveguide assignment, beamforming, and power control, which leads to a high-dimensional and non-convex optimization problem. \par

\subsection{Related Works}
Since pinching antennas were proposed by DOCOMO, many research papers have been conducted on pinching-antenna systems to show its superiority from both analytical and optimization perspectives. In \cite{dingpinchi}, a mathematical model for pinching-antenna systems was first developed, and analytical results were presented to demonstrate their superior performance in both downlink orthogonal multiple access (OMA) and non-orthogonal multiple access (NOMA) scenarios. Based on this, the upper bound on the array gain of pinching-antenna systems was analyzed in \cite{ouyangarraygain}.   To further extend PA research to practical settings, the authors in \cite{tyroperformance} derived closed-form expressions for the outage probability and average rate of downlink pinching-antenna systems by incorporating both free-space path loss and waveguide attenuation under realistic conditions. In addition to downlink scenarios, \cite{houperformance} derived closed-form expressions for the analytical, asymptotic, and approximate ergodic rates in three uplink configurations: multiple PAs serving a single user, a single PA serving a single user, and a single PA serving multiple users. Beyond analytical performance studies, another important research direction has focused on optimization methods for PA systems, with the goal of jointly designing PA placement and resource allocation to maximize system performance. In \cite{tegosuplinkrate, xudownlinkrate}, the uplink and downlink data rates were maximized by optimizing PA placement. These works were further extended to multi-user NOMA scenarios, where the sum rate maximization problems for uplink and downlink systems were investigated in \cite{zeng2025sum, zhousumdownlink}. Furthermore, the authors in \cite{xiepinching} proposed a low-complexity PA placement design that achieves near-maximal sum rate without relying on computationally intensive algorithms. However, the aforementioned studies require frequent repositioning of PAs when users change locations, which poses significant hardware challenges. To overcome this limitation, an antenna activation mechanism was proposed in \cite{wangactive1}, eliminating the need to physically move PAs. In this approach, multiple PAs are pre-deployed along the waveguide, and their activation status is controlled such that only the required PAs are activated while the others remain inactive. Then, the joint optimization of antenna activation and resource allocation to maximize sum rate in a multi-waveguide pinching-antenna system was studied in \cite{wangactive2}. \par
Beyond studies on the pinching-antenna system itself, recent works have investigated its integration with 6G technologies, particularly integrated sensing and communication (ISAC). In \cite{ding2025pinching}, the Cramér–Rao lower bound (CRLB) was derived for PA-assisted ISAC, showing notable localization accuracy gains over conventional antennas. A two-waveguide design was later proposed in \cite{zhang2025integrated}, employing penalty-based alternating optimization to boost illumination power under QoS constraints. To enhance sensing robustness, \cite{khalili2025pinching} exploited target diversity via dynamic PA activation and a Chernoff-bound based convex approximation method. Most recently, \cite{mao2025multi} addressed multi-waveguide PA-ISAC systems, jointly optimizing PA placement and beamforming through fine-tuning and successive convex approximation. Meanwhile, physical layer security (PLS) has emerged as another important application domain for pinching-antenna systems. In \cite{zhu2025pinching}, a secure framework with pinching beamforming was proposed, where a PA-wise successive tuning algorithm enhanced secrecy in single-waveguide systems and artificial noise was employed for multi-waveguide setups. An artificial-noise aided scheme was further developed in \cite{papanikolaou2025secrecy}, combining closed-form single-waveguide solutions with alternating optimization for multiple waveguides, achieving clear secrecy gains over conventional MIMO. \par
\begin{figure}[t]
     \centering
     \includegraphics[width=0.5\textwidth]{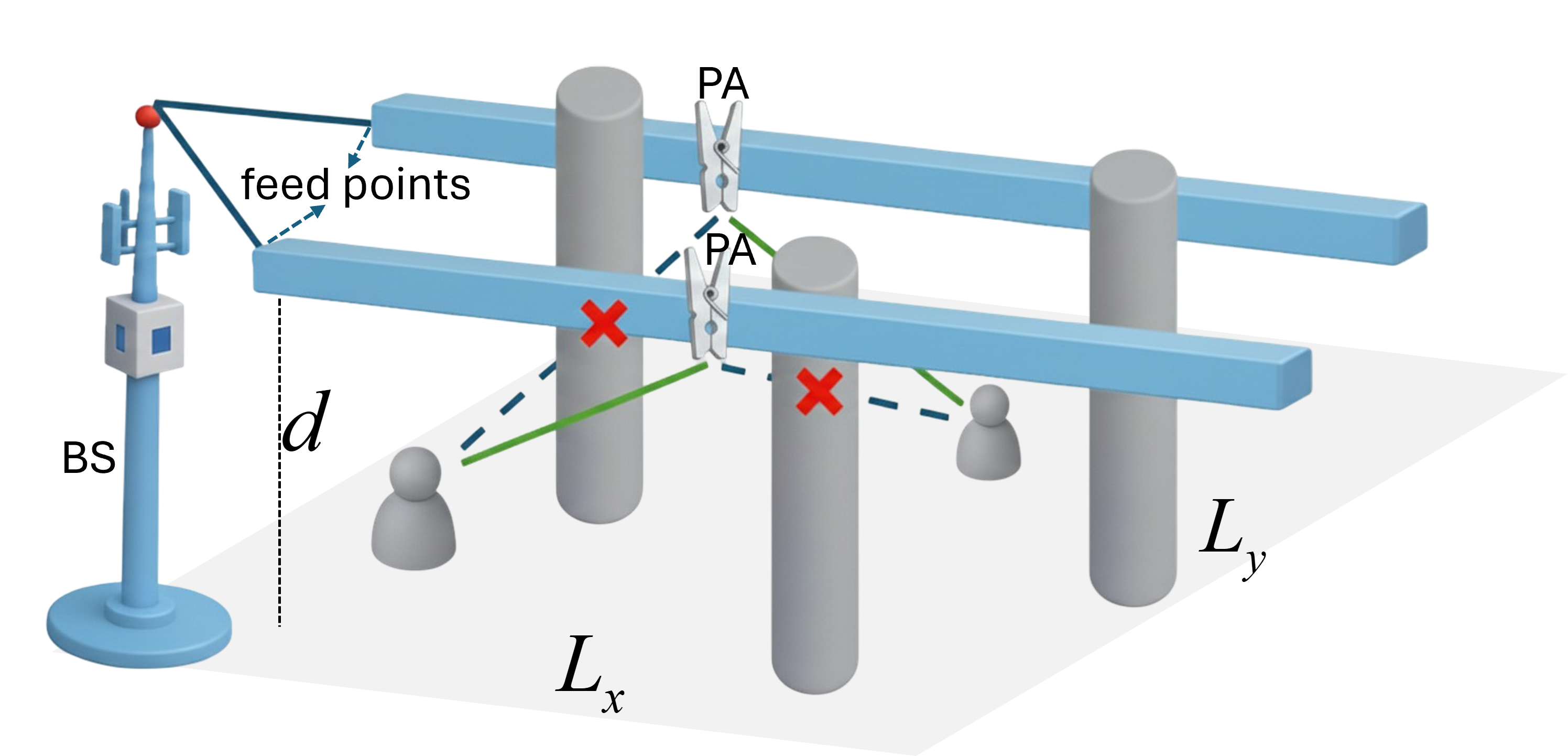}
     \caption{Pinching-antenna systems in blockage-aware environments.}
     \label{system model}
     \vspace{-0.5cm}
\end{figure}
\begin{figure*}[t]
     \centering
     \includegraphics[width=0.95\textwidth]{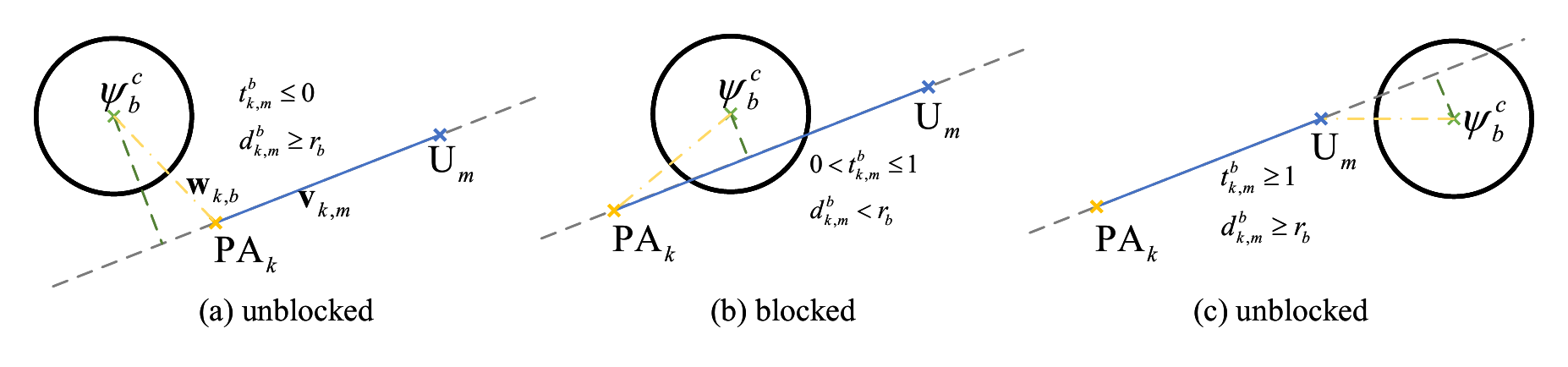}
     \caption{Three cases of the blockage model.}
     \label{blockage case}
     \vspace{-0.5cm}
\end{figure*}
\subsection{Motivation and Contributions}
The aforementioned studies primarily assumed ideal environments without obstacles. However, in practical scenarios, especially indoor environments, obstacles are common and often block LoS links. A recent study \cite{dingblockage} introduced pinching-antenna systems into blockage-aware environments and revealed a novel feature: PAs can exploit obstacles to bypass co-channel interference. A following work \cite{wangblockage} maximized the sum rate of pinching-antenna systems in blockage-aware environments. However, both studies relied on statistical blockage models, which are insufficient to accurately capture the characteristics of realistic blockage-aware environments. To further demonstrate the potential of pinching-antenna systems in practical blockage-aware environments, this paper develops a realistic blockage model for cylinder-shaped obstacles, which abstract structural elements such as pillars and columns that are common in large indoor venues including shopping malls, office buildings, factories and warehouses. A special case is first investigated, where each PA serves a single user and its position is selected from discrete candidate points along the waveguide. Subsequently, a general case is studied, where each PA serves all users with continuously adjustable placement along the waveguide. The main contributions of this paper are summarized as follows:
\begin{itemize}
    \item We propose a deterministic blockage model for cylinder-shaped obstacles. The proposed model checks whether the LoS link is available by checking the closest point from the obstacle’s center to the PA–user line segment and comparing that distance to the obstacle’s radius. This yields exact blockage regions bounded by tangents, produces per-PA blockage maps without any stochastic approximation, and plugs directly into the link model through a binary LoS indicator.
    \item We consider a special case where each PA serves a single user and chooses from predefined positions on its associated waveguide in blockage-aware environments. First, we perform a one-to-one waveguide–user assignment via the Hungarian algorithm. Next, we propose a surrogate-assisted block-coordinate
 (BCD) algorithm to efficiently search PA positions.  
    \item  We further consider a general case where each PA serves all users and can be positioned anywhere along the waveguide in blockage-aware environments. The problem is formulated as a joint optimization of beamforming and PA placement. To address non-smooth LoS transitions caused by obstacles, a weighted minimum mean square error integrated deep deterministic policy gradient (WMMSE-DDPG) reinforcement learning scheme is proposed to jointly optimize beamforming and PA placement.
    \item  Simulation results show that the proposed algorithms consistently outperform benchmark schemes in both cases. Moreover, the results demonstrate that pinching-antenna systems can exploit obstacles to mitigate co-channel interference, transforming potential blockages into performance gains while conventional fixed-antenna systems cannot achieve this advantage.
\end{itemize}

\subsection{Organization}
The remainder of this paper is organized as follows. Section II introduces the blockage model and system model. Sections III and IV present the proposed algorithms for the special and general cases, respectively. Section V provides numerical results. Finally, Section VI concludes the paper.

\section{System Model}
In this section, we first build a blockage model for cylinder-shaped obstacles in pinching-antenna systems. The potential blockage area and whether the user is located in a blockage area can be determined by the blockage model. Then, we integrate the blockage model into the system model to indicate how blockage affects communication performance. In this paper, we consider a system where $K$ parallel waveguides are deployed within a rectangular service area of dimensions $L_x \times L_y$ at a uniform height $d$. The environment contains $B$ cylinder-shaped obstacles, each with height $d$. A single PA is deployed on each waveguide, with the PA on the $k$-th waveguide denoted by ${\rm PA}_k$. Moreover, $M$ single-antenna users, denoted by ${\rm U}_m$ for $m = 1, \ldots, M$, are uniformly distributed across the service area. The overall system layout is illustrated in Fig. \ref{system model}.

\subsection{Blockage Model}
To build the blockage model, we consider the projections of waveguides and obstacles on the two-dimensional $x-y$ plane. To efficiently design the PA placement, we assume that a user is located within the blockage area and receives no service if its LoS link is blocked. Scattering and multi-path effects are not considered in this analysis, since they will make the blockage area stochastic and spatially uncertain. \par

Note that the projection of a cylinder-shaped obstacle on the $x-y$ plane is a circle. Let $\psi_b^c = (x_b^c,y_b^c)$ and $r_b$ denote the center coordinate and the radius of the $b-$th cylinder-shaped obstacle's projection, respectively. The coordinate of ${\rm PA}_k$ is $\psi^{\rm Pin}_k = (x_k^{\rm p}, y_k^{\rm p})$, and the coordinate of ${\rm U}_m$ is $\psi_m = (x_m,y_m)$, respectively. We draw two tangents from $\psi_k^{\rm Pin}$ to the circular projection. The blockage region of ${\rm PA}_k$ caused by the $b-$th obstacle is defined as the area enclosed by the two tangents, the boundary of the service area, and the circular projection, denoted by $\mathcal{A}_k^b$. To determine whether ${\rm U}_m$ is located in $\mathcal{A}_k^b$, we calculate the minimum distance from the center of the obstacle to the line segment connecting the PA and the user. A blockage occurs if this distance is less than the obstacle's radius. Let $\mathbf{v}_{k,m} = \psi_m - \psi_k^{\rm Pin} = [x_m - x_k^{\rm p}, y_m - y_k^{\rm p}]$ denote the vector from the ${\rm PA}_k$ to ${\rm U}_m$ and $\mathbf{w}_{k,b} = \psi_b^c - \psi_k^{\rm Pin} = [x_b^c - x_k^{\rm p}, y_b^c - y_k^{\rm p}]$ denote the vector from ${\rm PA}_k$ to the $b-$th obstacle's center, respectively. The first step is to calculate the normalized projection of $\mathbf{w}_{k,b}$ onto $\mathbf{v}_{k,m}$, which is given by
\begin{equation}
    t_{k,m}^b = \frac{\mathbf{w}_{k,b} \cdot \mathbf{v}_{k,m}}{||\mathbf{v}_{k,m}||^2}. \label{projection parameter}
\end{equation}

Let $\Tilde{l}_{k,m}$ denote the line segment from ${\rm PA}_k$ to ${\rm U}_m$ and $l_{k,m}$ denote the infinite line containing the line segment $\Tilde{l}_{k,m}$. $t_{k,m}^b$ determines the location of the point on $l_{k,m}$ that is closest to the $b-$th obstacle's center $\psi_b^c$. Let $\Tilde{\psi}_{k,m}^b$ denote the closest point on $\Tilde{l}_{k,m}$ to $\psi_b^c$. $\Tilde{\psi}_{k,m}^b$ can be calculated as follows:
\begin{equation}
     \Tilde{\psi}_{k,m}^b = \begin{cases}
        \psi_k^{\rm Pin}, & \mbox{if}~~t_{k,m}^b \leq 0 \\
        \psi_k^{\rm Pin} + t_{k,m}^b \mathbf{v}_{k,m}, & \mbox{if}~~ 0< t_{k,m}^b <1 \\
        \psi_m, & \mbox{if}~~t_{k,m}^b \geq 1
    \end{cases}. \label{closest point}
\end{equation}
Then, the distance between $ \Tilde{\psi}_{k,m}^b$ and $\psi_b^c$, denoted by $d_{k,m}^b$ can be calculated by
\begin{equation}
    d_{k,m}^b = ||\psi_b^c -  \Tilde{\psi}_{k,m}^b||. \label{minimum distance}
\end{equation}
If $d_{k,m}^b \leq r_b$, the LoS link between ${\rm PA}_k$ and ${\rm U}_m$ is blocked by the $b$-th obstacle, otherwise $d_{k,m}^b > r_b$.
\begin{prop}
    The LoS link between ${\rm PA}_k$ and ${\rm U}_m$ can only be blocked by the $b$-th obstacle when $ 0< t_{k,m}^b <1$. \label{prop 1}
\end{prop}
\begin{proof}
    There are two physical assumptions. The first one is that PAs are outside the obstacle, which means $||\psi_b^c - \psi_k^{\rm Pin}|| > r_b$ and the second one is that users are also outside the obstacle, which means $||\psi_b^c - \psi_m|| > r_b$. From \eqref{closest point}, when $t_{k,m}^b <0$, we have  $\Tilde{\psi}_{k,m}^b = \psi_k^{\rm Pin}$. This directly leads to $d_{k,m}^b = ||\psi_b^c -  \psi_k^{\rm Pin}|| > r_b$, ensuring the LoS link is not blocked when $t_{k,m}^b \leq 0$. A similar proof for the case where $t_{k,m}^b \geq 1$ can be derived from the second physical assumption. 
\end{proof}
Given Proposition \ref{prop 1}, the condition that the LoS link between ${\rm PA}_k$ and ${\rm U}_m$ is blocked by the $b$-th obstacle can be expressed as 
\begin{equation}
    d_{k,m}^b \leq r_b ~~\text{AND} ~~ 0<t_{k,m}^b<1. \label{LoS blockage conditions}
\end{equation}
Fig. \ref{blockage case} illustrates three cases of the blockage model. In this figure, $\mathbf{w}_{k,b}$ and $\mathbf{v}_{k,m}$ are denoted by the yellow dashed line segment and blue solid line segment, respectively. \par

The total blockage region of ${\rm PA}_k$ caused by all obstacles is expressed as
\begin{equation}
    \mathcal{A}_k = \bigcup_{b=1}^B \mathcal{A}_k^b. \label{total blockage region}
\end{equation}
Fig. \ref{blockage model1} illustrates a simple case consisting of two PAs and one obstacle. 
\begin{figure}[t]
     \centering
     \includegraphics[width=0.4\textwidth]{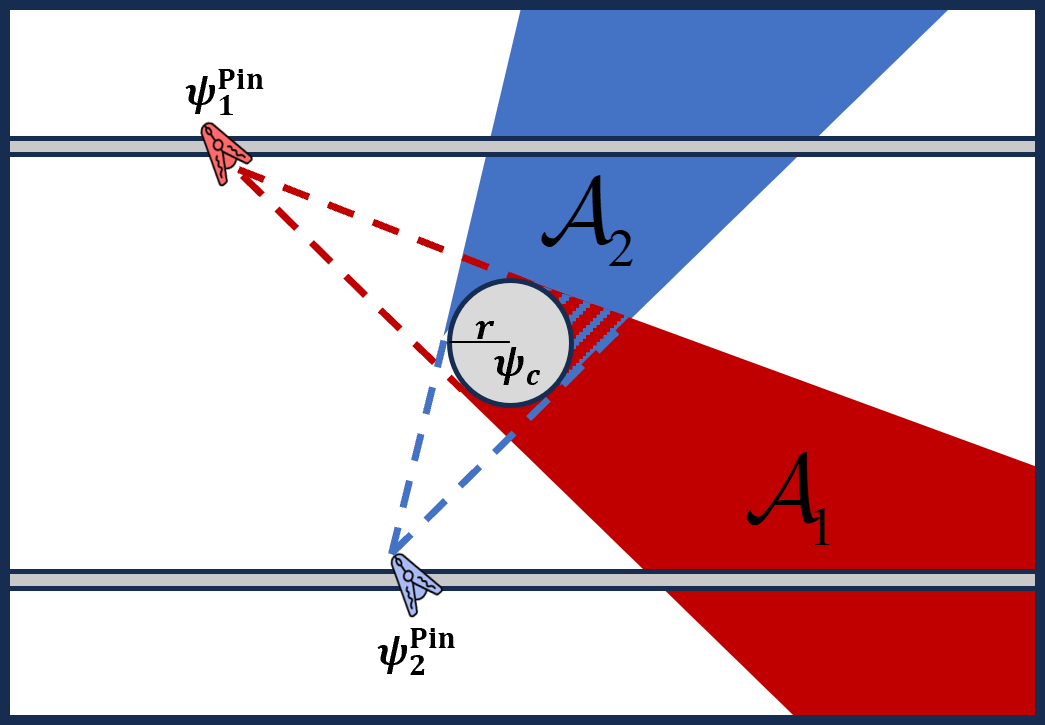}
     \caption{Illustration of the blockage from a cylinder-shaped obstacle.}
     \label{blockage model1}
\end{figure}

\subsection{Transmission Model and Problem Formulation}
Similar to \cite{dingblockage}, if ${\rm U}_m$ is located in ${\rm PA}_k$'s blockage area $\mathcal{A}_k$, the channel gain between ${\rm U}_m$ and ${\rm PA}_k$ is assumed to be $0$, otherwise, it is $h_{k,m} = \frac{\sqrt{\eta} e^{-2\pi j\left(\frac{1}{\lambda} |\psi_m - \psi_k^{\rm Pin}| + \frac{1}{\lambda_g}|\psi_{k,0}^{\rm Pin} - \psi_k^{Pin}|\right)}}{|\psi_m - \psi_k^{\rm Pin}|}$. $\eta = \frac{c^2}{16 \pi^2 f_c^2}$ denotes the free-space path loss coefficient, where $c$ is the speed of light and $f_c$ is the carrier frequency. $\lambda_g = \frac{\lambda}{n_{\rm eff}}$ denotes the waveguide wavelength in a dielectric waveguide, where $n_{\rm eff}$ denotes the effective refractive index of a dielectric waveguide. $\psi_{k,0}^{\rm Pin}$ denotes the position of the feed point of the $k-$th waveguide. Therefore, the channel gain between ${\rm U}_m$ and ${\rm PA}_k$ can be expressed as $\Tilde{h}_{k,m} = \alpha_{k,m} h_{k,m}$, where $\alpha_{k,m}$ is an indicator function for the LoS blockage and is given by
\begin{equation}
    \alpha_{k,m} = \begin{cases}
        0, & \mbox{if}~~\psi_m \in \mathcal{A}_k \\
        1, & \mbox{if}~~\psi_m \not\in \mathcal{A}_k
    \end{cases}. \label{LoS blockage indicator}
\end{equation} \par
Recall that the signal fed into the same waveguide must be the same in pinching-antenna systems \cite{dingpinchi}. Therefore, each waveguide transmits the superimposed signal of all users to serve them simultaneously. The superimposed signal transmitted by the $k$-th waveguide can be expressed as follows:
\begin{equation}
    s_k = \sum\limits_{m=1}^M p_{k,m} s_m, \label{design2 signal transmitted by k waveguide}
\end{equation}
where $p_{k,m}$ denotes the beamforming coefficient assigned to ${\rm U}_m$ on the $k$-th waveguide and $s_m$ denotes the desired signal of ${\rm U}_m$.  Similar to \cite{dingpinchi}, ${\rm U}_m$ receives signals from all the waveguides and its observation is given by
\begin{equation}
    y_m = \sum\limits_{k=1}^K \Tilde{h}_{k,m} p_{k,m}  s_m + \sum\limits_{i \neq m}\sum\limits_{k=1}^K \Tilde{h}_{k,m} p_{k,i}  s_i + n_m, \label{design 2 received signal}
\end{equation}
where $n_m$ denotes the additive noise with power $\sigma^2$. We assume that the signal satisfies $\mathbb{E}(|s_m|^2) = 1, \forall m$, where $\mathbb{E}(\cdot)$ is the expectation operation. Therefore, ${\rm U}_m$'s  signal-to-interference-plus-noise ratio (SINR) can be expressed as follows:
\begin{equation}
    \text{SINR}_m = \frac{|\sum_{k=1}^K \Tilde{h}_{k,m} p_{k,m}|^2}{\sum_{i \neq m} |\sum_{k=1}^K \Tilde{h}_{k,m} p_{k,i}|^2 + \sigma^2}. \label{design 2 SINR for user m}
\end{equation}
\eqref{design 2 SINR for user m} can be expressed as a compact form, which is given by
\begin{equation}
    \text{SINR}_m = \frac{|\Tilde{\mathbf{h}}_m^H \mathbf{p}_m|^2}{\sum_{i \neq m} |\Tilde{\mathbf{h}}_m^H \mathbf{p}_i|^2 + \sigma^2 }, \label{design 2 SINR for user m 2}
\end{equation}
where $\Tilde{\mathbf{h}}_m = [\Tilde{h}_{1,m}, \cdots, \Tilde{h}_{K,m}]^T$, and $\mathbf{p}_m = [p_{1,m}, \cdots, p_{K,m}]^T$. 
Then, the data rate of ${\rm U}_m$ is calculated by 
\begin{equation}
    R_m = \log_2(1 + \text{SINR}_m), \label{design 1 data rate user m}
\end{equation}
and the sum rate maximization problem can be formulated as follows:
\begin{subequations}\label{II-Prob0} 
\begin{align}
{\rm P_{0}}: \quad &\max_{\{\mathbf{\Psi},\mathbf{P}\}} \sum\limits_{m=1}^M R_m \label{PII00}\\
\text{s.t.} \quad & R_m \geq R_t, \quad \forall m \label{PII01}\\
\quad\quad & \sum_{m=1}^M ||\mathbf{p}_m||^2 \leq P_t,\label{PII02}\\
\quad\quad & 0 \leq x_k^{\rm p} \leq L_x, \quad\forall k \label{PII03}
\end{align}
\end{subequations}
where $\Psi = [\psi_1^{\rm Pin}, \cdots, \psi_K^{\rm Pin}]$ denotes the PA placement vector and $\mathbf{P} = [\mathbf{p}_{1}, \cdots, \mathbf{p}_M]$ denotes the beamforming matrix. Constraint \eqref{PII01} guarantees each user's data rate should meet the minimal target data rate requirement $R_t$. Constraint \eqref{PII02} ensures that the total transmit power does not exceed the power budget $P_t$. Constraint \eqref{PII03} restricts that the PA has to be placed on the waveguide.  Note that the channel vector $\Tilde{\mathbf{h}}_m$ is determined by the PA placement, hence $\Psi$ and $\mathbf{P}$ are coupled in the objective function \eqref{PII00} and the constraint \eqref{PII01}. As a result, ${\rm P_{0}}$ is difficult to solve directly.  

\section{A Special Case for the Discrete PA Placement Design}
In this section, we consider a special case which can provide more insights. In this case, each PA is restricted to a set of discrete positions along the waveguide and serves a single user exclusively. Specifically, each waveguide is uniformly divided into $N$ candidate positions, and the PA can only be placed at these positions. The set of $x$-coordinates for all candidate locations is given by $\mathcal{X} = \{\frac{n L_x}{N}|n=1,\cdots,N\}$. Moreover, each user is assigned to a waveguide and the PA on the waveguide only serves this user. We assume that each waveguide equally shares the power budget, hence, the transmit power of each waveguide is $P = \frac{P_t}{K}$. The signal transmitted by the $k$-th waveguide can be expressed as follows:
\begin{equation}
    s_k = \sum\limits_{m=1}^M \beta_{k,m} \sqrt{P} s_m, \label{design 1 signal}
\end{equation}
where $\beta_{k,m} \in \{0,1\}$ denotes the waveguide assignment indicator. Specifically, $\beta_{k,m}=1$ indicates that the $m$-th user is assigned the $k$-th waveguide, and $\beta_{k,m}=0$ otherwise. ${\rm U}_m$'s received signal can be expressed as follows:
\begin{equation}
    y_m = \sum\limits_{k=1}^K \Tilde{h}_{k,m} \sum\limits_{i=1}^M \beta_{k,i} \sqrt{P} s_i + n_m. \label{design 1 received signal}
\end{equation}
Therefore, ${\rm U}_m$'s SINR can be expressed as follows:
\begin{equation}
    \text{SINR}_m = \frac{P\sum_{k=1}^K \beta_{k,m} |\Tilde{h}_{k,m}|^2 }{P\sum_{i \neq m} \sum_{k=1}^K \beta_{k,i}|\Tilde{h}_{k,m} |^2 + \sigma^2}. \label{design 1 SINR for user m}
\end{equation}
Then, the data rate of ${\rm U}_m$ can be calculated by \eqref{design 1 data rate user m}.\par

In this special case, the sum rate maximization problem can be formulated as follows:
\begin{subequations}\label{I-Prob0} 
\begin{align}
{\rm P_{s:0}}: \quad &\max_{\{\mathbf{\Psi}, \boldsymbol{\beta}\}} \sum\limits_{m=1}^M R_m \label{PI00}\\
\text{s.t.} \quad & R_m \geq R_t, \quad \forall m \label{PI01}\\
\quad\quad & \sum\limits_{k=1}^K \beta_{k,m} = 1, \quad \forall m \label{PI02}\\
\quad\quad & \sum\limits_{m=1}^M \beta_{k,m} = 1, \quad \forall k \label{PI03}\\
\quad\quad & \beta_{k,m} \in \{0,1\}, \quad \forall k,m \label{PI04}\\
\quad\quad & x_k^{\rm p} \in \mathcal{X}, \quad \forall k \label{PI05} 
\end{align}
\end{subequations}
where $\boldsymbol{\beta}$ collects all waveguide assignment indicators. Constraint \eqref{PI01} guarantees the QoS requirement of each user. Constraints \eqref{PI02} and \eqref{PI03} ensure that each user is associated with exactly one waveguide, and each waveguide serves only one user. Note that ${\rm P_{s:0}}$ is a mixed-integer non-convex problem, which is difficult to solve directly. Given the coupling between the two variables, alternating optimization provides an effective means to address the resulting multi-variable problem, where in each iteration one variable is held fixed while the other is optimized. To simplify the problem ${\rm P_{s:0}}$ and enable efficient optimization, we consider the case $K=M$, such that each waveguide is exclusively assigned to one user.

\subsection{Waveguide Assignment via the Hungarian Algorithm}
The PA placement is assumed to be fixed during waveguide assignment optimization. Given the PA locations, the channel  $\tilde h_{k,m}$ is fixed \footnote{At the initial stage, each PA is randomly placed on the waveguide. We use this random PA placement as a starting point to optimize waveguide assignment.}. An important observation is that, once a user is assigned to a specific waveguide, the interference it experiences from the remaining PAs remains unchanged, regardless of how the other users are allocated. As a result, the sum rate with one-to-one mapping reduces to a linear assignment. In this case, it can be efficiently solved by the Hungarian Algorithm. \par
The first step is to build a weight matrix. For user $\mathrm{U}_m$ served by ${\rm PA}_k$, define
\begin{equation}
r_{m,k}\triangleq \log_2\!\left(1+\frac{P\,|\tilde h_{k,m}|^2}{P\!\left(\sum_{k'=1}^{K}|\tilde h_{k',m}|^2-|\tilde h_{k,m}|^2\right)+\sigma^2}\right),
\label{pair-rate}
\end{equation}
and collect $\mathbf{W}=[w_{m,k}]\in\mathbb{R}^{K\times K}$ with $w_{m,k}=r_{m,k}$ if the link is feasible and $w_{m,k}=-\infty$ otherwise, where feasibility means no LoS blockage $\alpha_{k,m}=1$. If any row/column of $\mathbf{W}$ has no feasible entry, the assignment is infeasible. Then, the assignment problem can be formulated as
\begin{algorithm}[t]
    \caption{Hungarian-based Waveguide Assignment}\label{hungarian}
    \begin{algorithmic}[1] 
        \STATE {\bf Row reduction:} Each row subtracts its row minimum in $\mathbf{C}$ to obtain $\mathbf{C}_1$.
        \STATE {\bf Column reduction:} Each column subtracts its column minimum in $\mathbf{C}_1$ to obtain $\mathbf{C}_2$.
        \STATE {\bf Initial stars:} Star a maximal set of independent zeros in $\mathbf{C}_2$; cover columns containing a starred zero.
        \WHILE{ number of covered columns $< K$}
            \STATE \textbf{Prime/augment:} Find an \emph{uncovered} zero and prime it.
            \IF {the primed zero’s row has no starred zero}
                \STATE build an alternating (prime$\leftrightarrow$star) path, flip marks (prime$\to$star, star$\to$unstar), clear all primes, uncover all rows/columns, then cover columns of starred zeros.
            \ELSE 
                \STATE Cover that row and uncover the column containing the starred zero; continue searching for an uncovered zero.
            \ENDIF
        \STATE \textbf{If no uncovered zero exists:} Let $\delta$ be the smallest uncovered entry; subtract $\delta$ from all uncovered rows and add $\delta$ to all covered columns.
        \ENDWHILE
        \STATE \textbf{Output:} Starred-zero positions give the optimal one-to-one assignment; objective is $\sum w_{m,k}$ at starred entries.
    \end{algorithmic}
\end{algorithm}
\begin{subequations}\label{I-Prob1} 
\begin{align}
{\rm P_{s:1}}: \quad &\max_{\{\boldsymbol{\beta}\}} \sum_{m=1}^{K}\sum_{k=1}^{K}\beta_{k,m}\, w_{m,k} \label{PI10}\\
\text{s.t.} \quad & \sum\limits_{k=1}^K \beta_{k,m} = 1, \quad \forall m \label{PI11}\\
\quad\quad & \sum\limits_{m=1}^M \beta_{k,m} = 1, \quad \forall k \label{PI12}\\
\quad\quad & \beta_{k,m} \in \{0,1\}, \quad \forall k,m \label{PI13}
\end{align}
\end{subequations}
However, the Hungarian Algorithm is a minimization method, which cannot be directly applied to ${\rm P_{s:1}}$. To make the Hungarian algorithm applicable, we convert weight matrix $\mathbf{W}$ to a non-negative cost matrix $\mathbf{C}$ for minimization. The element-wise transformation is given by
\begin{equation}
c_{m,k}=
\begin{cases}
c_{\max}-w_{m,k}, & \text{if $w_{m,k}$ is feasible},\\
M_{\text{big}}, & \text{if $w_{m,k}=-\infty$ },
\end{cases}
\label{cost-transform}
\end{equation}
where $c_{\max}=\max\{\,w_{m,k}~|~w_{m,k}~\text{is feasible}\}$ and $M_{\text{big}}$ is a constant with $M_{\text{big}} \gg c_{\text{max}}$. Then, ${\rm P_{s:1}}$ can be converted to a minimization problem by replacing $w_{k,m}$ with $c_{k,m}$, which the Hungarian Algorithm is applicable to. \par
Following the standard terminology of the Hungarian algorithm, any zero in the cost matrix may be marked as \emph{starred} ($0^\star$) or \emph{primed} ($0'$). A starred zero represents the current tentative assignment: at most one star per row and per column, and every column containing a star is \emph{covered}. \emph{Covering} a row/column means marking it so it is ignored in subsequent searches. A primed zero is a temporary mark used while searching for an augmenting path; an element can never be both starred and primed. If a primed zero appears in a row with no star, we build an alternating path and flip the marks along it (primes $\to$ stars, stars $\to$ unstar), thereby increasing the number of stars; all primes are then erased.  The details of the Hungarian Algorithm are summarized in Algorithm \ref{hungarian}, where \emph{independent zeros} means no two starred zeros share the same row or column and \emph{maximal} means we cannot add another zero without breaking independence.

\begin{rem}
When $K>M$, some waveguides may be left idle. If idle waveguides are silent, the achievable rate of ${\rm U}_m$ assigned to the $k$-th waveguide depends on which other waveguides are active, which is given by
\begin{equation}
    R_m(k|\mathbf{y})=\log_2\!\left(1+\frac{P|\Tilde{h}_{k,m}|^2}{P\sum_{k'\neq k} y_{k'}|\Tilde{h}_{k',m}|^2 + \sigma^2}\right),
\end{equation}
where $y_k\in\{0,1\}$ indicates whether waveguide $k$ is active. Hence, the weight matrix $\mathbf{W}$ is dependent on the activity pattern $\mathbf y = [y_1, \cdots, y_K]$ rather than fixed. As a result, the Hungarian algorithm is not guaranteed optimal for sum rate.
\end{rem}

\subsection{PA Placement Problem Formulation}
The waveguide assignment is assumed to be fixed during PA placement optimization. Let $\Pi$ denote a mapping between users and waveguides. $\Pi(m) = k$ means the $m-$th user has been assigned to the $k-$th waveguide. For a given $\Pi$, \eqref{design 1 SINR for user m} can be recast into
\begin{equation}
    \text{SINR}_m = \frac{P |\Tilde{h}_{k,m}|^2 }{P\sum_{i \neq m} |\Tilde{h}_{\Pi(i),m} |^2 + \sigma^2}. \label{pa placement SINR for user m 1}
\end{equation}
Since the PA positions are discrete in this case, \eqref{pa placement SINR for user m 1} can be further recast into
\begin{equation}
    \text{SINR}_m = \frac{P \sum_{n=1}^N \gamma_{k,n} |\Tilde{h}_{k,m}^n|^2 }{P\sum_{i \neq m} \sum_{n=1}^N \gamma_{\Pi(i),n} |\Tilde{h}_{\Pi(i),m}^n |^2 + \sigma^2}. \label{pa placement SINR for user m 2}
\end{equation}
In \eqref{pa placement SINR for user m 2}, $\gamma_{k,n} \in \{0,1\}$ denotes the PA position indicator. Specifically, $\gamma_{k,n} = 1$ indicates ${\rm PA}_k$ is placed at $n-$th position on the $k-$th waveguide, i.e., $x_k^{\rm p} = \frac{n L_x}{N}$, and $\gamma_{k,n} = 0$ otherwise. $|\Tilde{h}_{k,m}^n|^2$ denotes the squared channel gain between ${\rm PA}_k$ and ${\rm U}_m$ when ${\rm PA}_k$ at the $n-$th position of the $k-$th waveguide, which is given by
\begin{equation}
    |\Tilde{h}_{k,m}^n|^2 = \frac{\alpha_{k,m}^n \eta }{(x_m - x_{k,n}^{\rm p})^2 + D_{k,m}}, \label{squared channel gain}
\end{equation}
where $D_{k,m} = (y_m - y_k^{\rm p})^2 + d^2$ is a geometry-dependent constant for the given user and waveguide, $\alpha_{k,m}^n$ and $x_{k,n}^{\rm p}$ are the LoS indicator and the $x$-coordinate of the $n-$th PA position on the $k-$th waveguide, respectively. Hence, for a fixed candidate $n$, $|h_{k,m}^n|^2$ is a known constant; across candidates, it varies only through the horizontal offset $x_m - x_{k,n}^{\rm p}$. \par
Note that the user's data rate is mainly determined by signal power and interference power. Let
\begin{equation}
    S_m = P \sum_{n=1}^N \gamma_{k,n} |\Tilde{h}_{k,m}^n|^2 \label{design 1 signal power}
\end{equation}
\begin{equation}
    I_m = P\sum_{i \neq m} \sum_{n=1}^N \gamma_{\Pi(i),n} |\Tilde{h}_{\Pi(i),m}^n |^2 \label{design 1 interference power}
\end{equation}
denote signal power and interference power of ${\rm U}_m$, respectively. To efficiently calculate $S_m$ and $I_m$, we build a power lookup matrix for each waveguide collecting all possible squared channel gains. The power matrix for the $k-$th waveguide can be expressed as follows:
\begin{equation}
    \mathbf{H}_k = \begin{bmatrix}
        |\Tilde{h}_{k,1}^1|^2 & \cdots  &|\Tilde{h}_{k,1}^N|^2 \\
        \vdots & \ddots & \vdots \\
        |\Tilde{h}_{k,M}^1|^2 & \cdots  &|\Tilde{h}_{k,M}^N|^2
    \end{bmatrix} \in \mathbb{R}^{M \times N}, \forall k \label{power matrix}.
\end{equation}
Precompute and store the power matrices so that, during PA placement updates, desired-signal and interference powers can be retrieved from lookup tables rather than recomputed. Then, \eqref{design 1 signal power} and \eqref{design 1 interference power} become the efficient lookup forms
\begin{equation}
    S_m = P\mathbf{e}_m^\top \mathbf{H}_k\boldsymbol{\gamma}_k \label{design 1 signal power lookup form}
\end{equation}
\begin{equation}
    I_m = P\sum_{i\neq m} \mathbf{e}_m^\top \mathbf{H}_{\Pi(i)}\,\boldsymbol{\gamma}_{\Pi(i)} \label{design 1 interference power lookup form}
\end{equation}    
where $\mathbf{e}_m$ is the $m$-th canonical basis vector and $\boldsymbol{\gamma}_k = [\gamma_{k,1}, \cdots, \gamma_{k,N}]$ is the PA position indicator vector. Hence, the data rate of ${\rm U}_m$ can be rewritten as
\begin{equation}
    R_m = \log_2 \left(1+\frac{S_m}{I_m+\sigma^2} \right). \label{design 1 rewritten data rate}
\end{equation}
With a fixed mapping $\Pi$, the PA placement optimization problem can be recast into
\begin{subequations}\label{prob:PA-gamma}
\begin{align}
{\rm P_{s:2}}: \max_{\{\boldsymbol\gamma_k\}}~& \sum_{m=1}^M \log_2\!\left(1+\frac{S_m}{I_m+\sigma^2}\right) \label{PI20}\\
\text{s.t.}~~&
R_m\ge R_t,\ \forall m \label{PI21} \\
& \mathbf e_{\Pi(k)}^\top \mathbf H_k\,\boldsymbol\gamma_k > 0,\ \forall k\quad \label{PI22}\\
&\sum_{n=1}^N \gamma_{k,n}=1,\  \forall k  \label{PI23} \\
& \gamma_{k,n}\in\{0,1\},\ \forall k,n \label{PI24}
\end{align}
\end{subequations}
Constraint \eqref{PI22} ensures that the served LoS link is not blocked. \par
This problem is non-convex, which makes direct optimization intractable. Moreover, changing the position of a single PA not only changes the desired-signal power for its served user but also changes the interference experienced by all other users. As a result, optimizing all PA positions jointly becomes highly complicated. To address this challenge, we adopt a BCD approach, which iteratively updates one PA’s position at a time while keeping others fixed. This decomposition simplifies the optimization and reduces computational complexity.\par
\subsection{Surrogate-Assisted Block Coordinate Discrete Search}
The key idea of BCD is to improve $\boldsymbol{\gamma}_k, \forall k$ by updating one waveguide at a time. Let ${\rm PA}_k$ currently locate at the candidate position $n$ and its served user be ${\rm U}_m$. If ${\rm PA}_k$ moves to another candidate position $n'$, the incremental signal and interference power updates via lookup in power matrix \eqref{power matrix} is given by 
\begin{equation}
    S_m' \gets P\,[\mathbf H_k]_{m,n'}, \label{self signal power update}
\end{equation}
\begin{equation}
    I_j' \gets I_j - P\,[\mathbf H_k]_{j,n} + P\,[\mathbf H_k]_{j,n'},\quad \forall j\neq m, \label{other interference power update}
\end{equation}
\begin{equation}
    I'_m \gets I_m, \label{self interference power update}
\end{equation}
\begin{equation}
    S'_j \gets S_j, \quad \forall j\neq m. \label{other signal power update}
\end{equation}
From the updating rule, moving a PA to a different position only affects the desired-signal power of its served user and the interference power experienced by the other users. In contrast, the interference power of the served user and the desired-signal power of the other users remain unchanged. Accordingly, the updated sum rate can be expressed as
\begin{equation}
    F' \triangleq \sum\limits_{m=1}^M \log_2 \left(1 + \frac{S'_m}{I'_m + \sigma^2}\right). \label{updated sum rate}
\end{equation}
For simplicity of notation, we denote the sum rate by $F$. The next step is to determine whether the new candidate position $n'$ is better than the original position $n$. The new candidate $n'$ is accepted if (i) $R_m' \geq R_t, \forall m$, and (ii) $F' > F$. In this case, accepted PA moves strictly increase $F$ with a finite state space. The BCD algorithm terminates in finitely several sweeps at a coordinate-wise solution. \par
Although the BCD algorithm can effectively solve the PA placement problem, it requires each PA to sequentially evaluate all candidate positions. Consequently, when $N$ is large (i.e., when each waveguide has many candidate locations), the computational complexity of the algorithm increases significantly. To avoid exact evaluation for all $N$ candidates, we propose a surrogate accelerated ranking mechanism. In this mechanism, a score is assigned to each move from the current position to a new position to quantify its contribution to the sum rate $F$. The candidate positions on the same waveguide are then ranked according to this score, and the PA only evaluates the top-ranked positions, while those with lower ranks are ignored. \par
Note that \eqref{design 1 rewritten data rate} can be rewritten as
\begin{equation}
    R_m = \frac{1}{\ln2} \left[\ln (S_m + I_m + \sigma^2) - \ln (I_m + \sigma^2)\right]. \label{design 1 rewritten data rate rewritten}
\end{equation}
Recall that when only $x_k^{\rm p}$ changes, only $S_m$ and $I_j, \forall j, j\neq m$ change. In order to describe the rate of change of $F$ with $S_m$ and every $I_j$, we calculate the partial derivatives of $F$ with $S_m$ and $F$ with every $I_j$. Let $T_m = S_m + I_m + \sigma^2$ and $U_m = I_m + \sigma^2$, then the partial derivatives are given by
\begin{equation}
    \frac{\partial F}{\partial S_m} = \frac{\partial R_m}{\partial S_m} = \frac{1}{\ln2} \cdot \frac{1}{T_m}, \label{partial of Sm}
\end{equation}
and
\begin{equation}
    \frac{\partial F}{\partial I_j} = \frac{\partial R_j}{\partial I_j} = \frac{1}{\ln2}\left(\frac{1}{T_j} - \frac{1}{U_j}\right), \forall j, j\neq m. \label{partial of Ij}
\end{equation}
Note that $\frac{\partial F}{\partial S_m} > 0$ and $\frac{\partial F}{\partial I_j} < 0$, which means that increasing the desired-signal power $S_m$ improves the sum rate, while increasing the interference power $I_j$ reduces it. Therefore, the candidate positions that provide stronger desired signals to their associated users and induce lower interference to other users are more favorable for sum rate maximization. \par
When ${\rm PA}_k$ moves from $n$ to $n'$, the variations of $S_m$ and $I_j$ are given by
\begin{equation}
    \Delta S_m = P([\mathbf{H}_k]_{m,n'} - [\mathbf{H}_k]_{m,n}) \label{variation of Sm}
\end{equation}
and
\begin{equation}
    \Delta I_j = P([\mathbf{H}_k]_{j,n'} - [\mathbf{H}_k]_{j,n}). \label{variation of Ij}
\end{equation}
Then, we calculate the first-order surrogate gain of $F$, which is given by
\begin{equation}
    \Delta \Tilde{F} = \left(\frac{1}{\ln2} \cdot \frac{1}{T_m}\right) \Delta S_m + \sum\limits_{\substack{j=1\\ j \neq m}}^M \left(\frac{1}{\ln2}\left(\frac{1}{T_j} - \frac{1}{U_j}\right)\right) \Delta I_j. \label{surrogate gain}
\end{equation}
\begin{algorithm}[t]
\caption{Surrogate-Assisted BCD for PA Placement}
\label{sbcd}
\begin{algorithmic}[1]
\STATE \textbf{Precompute:} $\{\mathbf H_k\}_{k=1}^K$.
\STATE \textbf{Initialize:} Choose feasible one-hot $\boldsymbol\gamma_k$ ($\sum_n\gamma_{k,n}=1$). Compute $S_i,I_i,R_i, \forall i, F$.
\FOR{$t=1$ to $T_{\max}$}
  \STATE improved = \textbf{false}
  \FOR{$k=1$ to $K$}
    \STATE $\Pi(m) = k$; current index $n = \arg\max_u \gamma_{k,u}$
    \STATE Compute weights $\zeta_m$ and $\theta_j$ for $j\neq m$
    \STATE For each feasible $n'$ with $[\mathbf H_k]_{m,n'}>0$, compute $\mathcal Q_k(n')$
    \STATE Build $\mathcal{N}_k$ by selecting top $N'$ candidates by descending $\mathcal Q_k(n')$
    \FOR{each $n' \in \mathcal{N}_k$ }
      \STATE Update $S_i',I_i', \forall i$ incrementally; compute $R_i', \forall i$ and $F'$
      \IF{$ R_i'\ge R_t, \forall i$ \AND $F'> F$}
        \STATE Accept: $\gamma_{k,n} = 0$, $\gamma_{k,n'} = 1$; set $S_i = S_i'$, $I_i = I_i'$, $R_i = R_i', \forall i$, $F = F'$
        \STATE improved = \textbf{true}; \textbf{break}
      \ENDIF
    \ENDFOR
  \ENDFOR
  \IF{improved == \textbf{false}} \STATE \textbf{break} \ENDIF
\ENDFOR
\STATE \textbf{Output:} $\{\boldsymbol\gamma_k\}$
\end{algorithmic}
\end{algorithm}
After some straightforward algebraic transformations, \eqref{surrogate gain} can be recast into
\begin{align}
    \Delta \Tilde{F} &= P\left(\zeta_m [\mathbf{H}_k]_{m,n'} + \sum\limits_{\substack{j=1\\ j \neq m}}^M \theta_j [\mathbf{H}_k]_{j,n'}\right) \notag \\
    &- P\left(\zeta_m [\mathbf{H}_k]_{m,n} + \sum\limits_{\substack{j=1\\ j \neq m}}^M \theta_j [\mathbf{H}_k]_{j,n} \right), \label{surrogate gain 2}
\end{align}
where $\zeta_m = \frac{1}{\ln2} \cdot \frac{1}{T_m}$ and $\theta_j = \frac{1}{\ln2}\left(\frac{1}{T_j} - \frac{1}{U_j}\right)$. Note that the first term is related to the candidate position $n'$ and the current position $n$ and the second term is only related to the current position $n$. Therefore, we use the first term as a score to evaluate the move from $n$ to $n'$. The score of the new candidate position $n'$ can be expressed as
\begin{equation}
    \mathcal{Q}_k(n') = \zeta_m [\mathbf{H}_k]_{m,n'} + \sum\limits_{\substack{j=1\\ j \neq m}}^M \theta_j [\mathbf{H}_k]_{j,n'}. \label{candidate score}
\end{equation}
All feasible candidate positions \footnote{A feasible position means this position satisfies constraint \eqref{PI21}.} are ranked according to \eqref{candidate score}, and the PA examines only the top $N'$ positions. Let $\mathcal{N}_k$ denote the set collecting the top $N'$ feasible positions' indices of the $k-$th waveguide. The optimized position of $x_k^{\rm p}$ is given by
\begin{equation}
    x_k^{\rm p*} = \frac{n^* L_x}{N}, n^* = \arg \max_{n^{'} \in \mathcal{N}_k} F. \label{optimized position}
\end{equation}
The details are summarized in Algorithm \ref{sbcd}. \textit{improved} is a Boolean variable to indicate if the new position improves the sum rate. $T_{\rm max}$ denotes the total number of sweeps.

\begin{rem}
The score in \eqref{candidate score} is linear in the precomputed power columns of $\mathbf H_k$. Thus all scores for waveguide $k$ can be computed by a single matrix–vector product, using only light operations (adds/multiplies rather than logs/divisions). After scoring, we select the top $N'$  candidates ($N'\!\ll\!N$) and evaluate the true objective $F$ only for these, rather than for all $N$  positions. This sharply reduces expensive evaluations and accelerates the algorithm.
\end{rem}

\begin{rem}
Compared with the matching-theoretic scheme in \cite{wangblockage}, the Hungarian-based waveguide–user assignment yields the globally optimal solution to the linear assignment problem when PA positions are fixed. Moreover, \cite{wangblockage} maximizes sum rate without per-user QoS guarantees, whereas the proposed surrogate-assisted BCD explicitly enforces the individual QoS constraint \eqref{PI21}.
\end{rem}

\section{Joint Optimization of Beamforming and Continuous PA Placement Design}
In this section, we focus on solving the optimization problem ${\rm P}_0$. The main difference from the special case is that each PA serves all users and the PA position is continuous on the waveguide. Since beamforming and PA placement are inherently coupled, we adopt an alternating optimization strategy to iteratively refine both variables.
\vspace{-0.1cm}
\subsection{Beamforming Design}
The PA placement is assumed to be fixed during beamforming optimization. Then, ${\rm P_0}$ becomes a classic sum rate maximization problem for beamforming design, which can be efficiently solved by the WMMSE algorithm. 
In the WMMSE algorithm, the scalar equalizer $u_m$, MSE $e_m$ and positive weight $w_m$ are defined as
\begin{equation}
    u_m = \frac{\Tilde{\mathbf{h}}_m^H \mathbf{p}_m}{\sum_{i=1}^M |\Tilde{\mathbf{h}}_m^H \mathbf{p}_i|^2 + \sigma^2}, \label{equalizer}
\end{equation}
\begin{equation}
    e_m = 1-2 \mathfrak{R}\{u_m^* \Tilde{\mathbf{h}}_m^H \mathbf{p}_m\} + |u_m|^2\left(\sum\limits_{i=1}^M |\Tilde{\mathbf{h}}_m^H \mathbf{p}_i|^2 + \sigma^2\right), \label{MSE}
\end{equation}
and
\begin{equation}
    w_m = \frac{1}{e_m}, \label{positive weight}
\end{equation}
respectively. According to \eqref{equalizer} and \eqref{MSE}, MSE can be rewritten as follows:
\begin{equation}
    e_m = \frac{1}{1+ \text{SINR}_m} \label{MSE rewritten}
\end{equation}
As a result, the QoS constraint \eqref{PII01} can be rewritten as follows:
\begin{equation}
    e_m \leq \delta_m, \label{QoS constraint}
\end{equation}
where $\delta_m = 2^{-R_t}$.
\begin{algorithm}[t]
\caption{WMMSE Algorithm for Beamforming}
\label{alg:wmmse}
\begin{algorithmic}[1]
\STATE \textbf{Initialize:} $\mathbf H$, feasible $\mathbf P^{(0)}$, $\lambda^{(0)} \geq 0$, $\mathbf N^{(0)}\succeq\mathbf 0$, $t\gets 0$
\REPEAT
    \STATE \textbf{Receiver/weight update}: compute $u_m^{(t+1)}$, $e_m^{(t+1)}$, $w_m^{(t+1)}$ via \eqref{equalizer},\eqref{MSE},\eqref{positive weight} from $\mathbf P^{(t)}$; then form $\mathbf U^{(t+1)},\mathbf W^{(t+1)}$
    \STATE \textbf{Primal update:} compute $\mathbf P^{(t+1)}$ by \eqref{eq:p-update} from $\mathbf U^{(t+1)},\mathbf W^{(t+1)}$, $\mathbf N^{(t)}, \lambda^{(t)}$
    \STATE \textbf{Dual updates:} update $\lambda^{(t+1)}$ and $\nu_m^{(t+1)}, \forall m$ using \eqref{gradient ascend lambda},\eqref{gradient ascend nu}; then form $\mathbf N^{(t+1)}$
    \STATE $t\gets t+1$
\UNTIL{convergence}
\STATE \textbf{Output:} $\mathbf P$
\end{algorithmic}
\vspace{-0.1cm}
\end{algorithm}
Then, the sum rate maximization problem is transformed into the equivalent weighted MSE minimization problem, which can be written as follows:
\begin{subequations}\label{II-Prob1} 
\begin{align}
{\rm P_{1}}: \quad &\min_{\{\mathbf{P}\}} \sum\limits_{m=1}^M (w_me_m - \log w_m) \label{PII10}\\
\text{s.t.} \quad & \; e_m \leq \delta_m, \quad \forall m\\
\quad\quad & \sum_{m=1}^M ||\mathbf{p}_m||^2 \leq P_t. \label{PII11}
\end{align}
\end{subequations}
With fixed $u_m, w_m, \forall m$, the subproblem ${\rm P_{1}}$ is convex quadratic. The Lagrange function of ${\rm P_{1}}$ is given by
\begin{align}
    \mathcal{L}  &=  \sum\limits_{m=1}^M (w_me_m - \log w_m) \notag \\
    &+ \sum\limits_{m=1}^M \nu_m (e_m - \delta_m) + \lambda \left(\sum\limits_{m=1}^M||\mathbf{p}_m||^2 - P_t \right), \label{Lagrange}
\end{align}
where $\nu_m, \forall m$ and $\lambda$ are non-negative Lagrange multipliers. Let $\mathbf{U} = \di\{u_1, \cdots, u_M\}$, $\mathbf{W} = \di \{w_1, \cdots, w_M\}$, and $\mathbf{N} = \di \{\nu_1, \cdots, \nu_M\}$. After some straightforward algebraic transformations. the $\mathbf{P}$-dependent Lagrange function can be expressed as follows:
\begin{align}
    \mathcal{L}(\mathbf{P}) &= \Tr (\mathbf{P}^H (\mathbf{A} + \lambda \mathbf{I}_K)\mathbf{P}) \notag \\
    & -\Tr ((\mathbf{H} \mathbf{U}(\mathbf{W} + \mathbf{N}))^H \mathbf{P}) - \Tr(\mathbf{P}^H \mathbf{H}\mathbf{U}(\mathbf{W} + \mathbf{N})), \label{Lagrange related to P}
\end{align}
where $\mathbf{I}_K$ is a $K \times K$ identity matrix, $\mathbf{A} = \mathbf{H}\mathbf{U}(\mathbf{W} + \mathbf{N})\mathbf{U}^H \mathbf{H}^H$ and $\mathbf{H} = [\Tilde{\mathbf{h}}_1, \cdots, \Tilde{\mathbf{h}}_M]$. According to the KKT conditions, we have
\begin{equation}
    \frac{\partial \mathcal{L}(\mathbf{P})}{\partial \mathbf{P}^*} = (\mathbf{A} + \lambda \mathbf{I}_K)\mathbf{P} - \mathbf{H}\mathbf{U}(\mathbf{W} + \mathbf{N}) = \mathbf{0}. \label{KKT condition}
\end{equation}
Then, the beamforming matrix can be calculated by
\begin{equation}
    \mathbf{P} = (\mathbf{A}+\lambda \mathbf{I}_K)^{-1}\mathbf{H}\mathbf{U}(\mathbf{W} + \mathbf{N}). \label{eq:p-update}
\end{equation}
According to complementary slackness, we have
\begin{equation}
    \lambda (\sum_{m} ||\mathbf{p}_m||^2 - P_t) = 0 \label{complementary slackness 1}
\end{equation}
and
\begin{equation}
    \nu_m (e_m - \delta_m) = 0, \forall m. \label{complementary slackness 2}
\end{equation}
We can update $\lambda$ and $\nu_m$ via gradient ascent. The updating rule is given by
\begin{equation}
    \lambda' \gets \left[\lambda+\tau(\|\mathbf P\|_F^2-P_t)\right]_+ \label{gradient ascend lambda} 
\end{equation}
and
\begin{equation}
    \nu_m' \gets \left[\nu_m + \rho (e_m - \delta_m)\right]_+, \forall m \label{gradient ascend nu} 
\end{equation}
respectively. The operation $[x]_+ = \max \{0,x\}$ ensures Lagrange multipliers are non-negative. The details are summarized in Algorithm \ref{alg:wmmse}.

\subsection{PA Placement Optimization}
Unlike open spaces where channel variations are smooth, obstacles cause abrupt changes in channel gains: a slight shift in PA position can suddenly switch a user from LoS to nLoS conditions. This discontinuity makes the optimization of PA placement challenging, since traditional gradient-based or convex approaches assume smooth objective functions. In this case, we therefore adopt DDPG as a solver to determine PA positions. \par
To apply DDPG, we first define the state, action and reward as follows:
\begin{itemize}
    \item State: The state space $s$ is a 1-D vector denoted by $s=(\psi_1, \cdots \psi_M, \psi_1^c, \cdots, \psi_B^c, r_1, \cdots, r_B)$. The state encapsulates information about user locations and obstacle characteristics.
    \item Action: The action space $a$ is the PA placement vector. Since each PA is deployed on one waveguide and the waveguides are fixed, the action can only be the $x$-axis coordinate of each PA, denoted by $a = [x_1^{\rm Pin}, \cdots, x_K^{\rm Pin}]$.
    \item Reward: Our goal is to maximize the sum rate. As a result, the intuitive reward function is the sum rate. However, to ensure the QoS constraint, we modify the reward function as follows
    \begin{equation}
       rd =  \sum\limits_{m=1}^MR_m - \sum\limits_{m=1}^M \varrho_m g_m, \label{reward function}
    \end{equation}
    where $\varrho_m$ and $g_m$ are the penalty weight and the violation score for the QoS constraint of ${\rm U}_m$, respectively. The violation score can be approximated by a smooth function, which is given by
    \begin{equation}
        g_m = \tau \log\left(1 + e ^{(R_t - R_m)/\tau}\right), \label{violation score}
    \end{equation}
    where $\tau$ is the temperature that controls how sharply the penalty transitions around the QoS boundary $R_t$. 
\end{itemize} \par
\begin{figure}[t]
     \centering
     \includegraphics[width=0.5\textwidth]{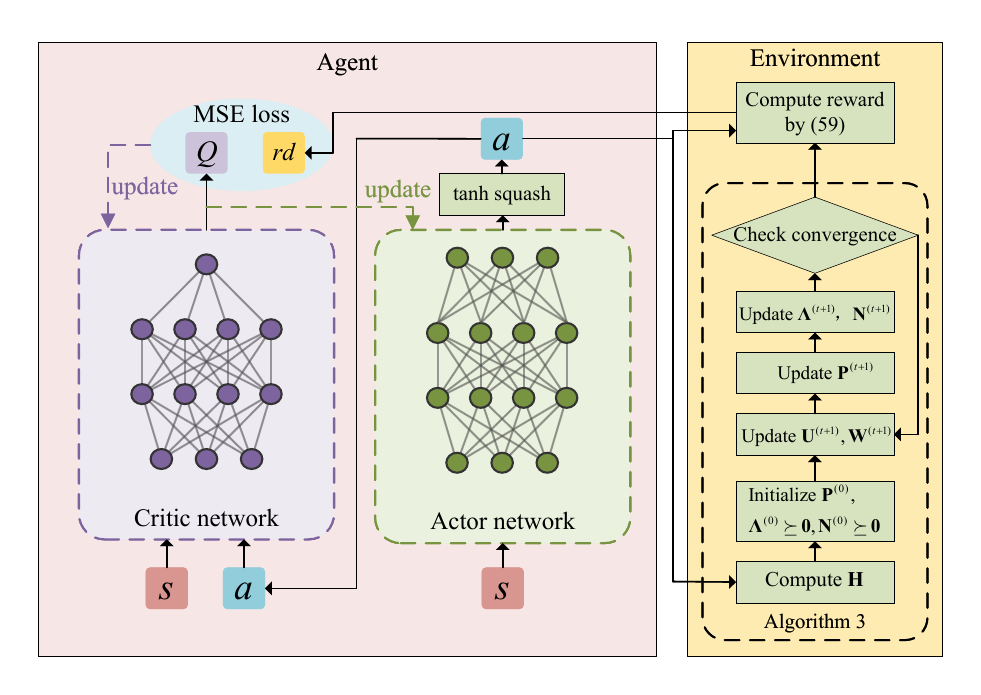}
     \caption{Block diagram of WMMSE-DDPG}
     \label{WMMSE-DDPG}
\end{figure}
Note that the PA placement design can be modeled as a contextual bandit problem. In particular, the agent observes the environment, selects PA positions, and immediately receives a reward based on the sum rate achieved for a given user and obstacle configuration. The objective is to maximize this one-step reward and there is no future reward to estimate. As a result, DDPG’s target networks do not contribute to learning in this setting. Accordingly, we train only an online actor–critic: one actor that outputs PA positions and one critic that regresses to the immediate reward and provides a stable learning signal. Let the online actor be $\pi(\phi)$ and the online critic be $Q(\theta)$, where $\phi$ and $\theta$ denote the corresponding parameters. The online critic network updates parameters by minimizing an MSE loss, which is given by
\begin{equation}
    \mathcal{L}(\theta) = \left(Q(s,a|\theta) - rd\ \right)^2. \label{MSE loss}
\end{equation}
The gradient of $\phi$ for updating the actor is calculated by
\begin{equation}
    \nabla_{\phi} J(\phi) =  \nabla_{\pi(s|\phi)} Q\left(s, \pi(s|\phi)\right)|\nabla_{\phi} \pi_{\phi}(s). \label{actor update}
\end{equation}
Since the actor outputs unconstrained logits $\mathbf{u} \in \mathbb{R}^K$, we need to regularize $\mathbf{u}$ to obtain a feasible action. The regularization is given by
\begin{equation}
    x_i^{\rm Pin} = \frac{L_x}{2}\left(\tanh(u_i) + 1 \right), i = \{1, \cdots, K\}, \label{action regularization}
\end{equation}
where $u_i$ is the $i-$th element of $\mathbf{u}$.

\subsection{Algorithm}
Fig. \ref{WMMSE-DDPG} presents the overall architecture of the proposed WMMSE-DDPG framework for joint optimization of continuous PA placement and beamforming. The agent receives the state, which contains user location information, and the actor outputs the corresponding PA coordinates as the action. The environment then constructs the obstacle-aware channel, executes the WMMSE algorithm to obtain the beamforming matrix, and computes the reward. The critic takes the state–action pair as input and estimates its value. It is trained using a mean-squared-error loss to align the predicted value with the observed reward, while the actor is subsequently updated through the critic’s feedback to improve future decisions.

\section{Simulation Results}
\begin{table}[t]
\centering
\caption{Simulation Parameters}
\begin{tabular}{lc}
\hline
\textbf{Parameter} & \textbf{Value} \\ \hline
Area size & $L_x{=}30$~m, $L_y{=}20$~m \\
Height & $d = 2.5$~m \\
Carrier frequency & $f_c{=}28$~GHz \\
Speed of light & $c{=}3\times10^8$~m/s \\
Noise power & $\sigma^2 = -120$~dBm/Hz \\
Transmit power & $P_t = 30$~dBm \\
PA candidates per waveguide & $N{=}100$ \\
Shortlist size in BCD & $N'{=}20$ \\
Target rate threshold & $R_t{=}0.5$~bps/Hz \\
Actor hidden sizes & $(256,256)$ \\
Critic hidden sizes & $(256,256)$ \\
Learning rate & $10^{-4}$ \\
Temperature & $\tau = 0.01$ \\ \hline
\end{tabular}
\label{tab:hyperparams_design1}
\end{table}
We consider a rectangular service area of size $L_x \times L_y = 30~\text{m} \times 20~\text{m}$ with $K$ parallel dielectric waveguides deployed at height $d = 2.5$~m. Within the service area, there are several cylinder-shaped obstacles. The key simulation parameters are given in Table~I. 
\subsection{Evaluation of the special case}
In this case, each waveguide is uniformly quantized into \(N=100\) candidate PA positions. The proposed algorithm is evaluated against several benchmarks under identical simulation settings. The considered methods are summarized as follows. \textit{BCD-AO}: the proposed alternating optimization framework that integrates the Hungarian algorithm for waveguide–user assignment with a surrogate-assisted BCD search for PA placement. \textit{SwapMatching}: a benchmark that employs the swap-matching mechanism presented in \cite{wangblockage}. \textit{RandomClosest}: users are randomly assigned to waveguides, and each PA is positioned at the candidate location closest to its associated user. \textit{HungarianRandom}: waveguide–user assignment is performed via the Hungarian algorithm, while PA positions are randomly selected from the candidate set. \textit{FixAntenna}: a conventional fixed-antenna configuration, where each antenna is placed at the feed point of its corresponding waveguide. \textit{RandomRandom}: both waveguide–user assignment and PA placement are randomly determined. The simulation results compare the achievable sum rate and feasibility ratio among these methods.

\begin{figure}[t]
     \centering
     \includegraphics[width=0.4\textwidth]{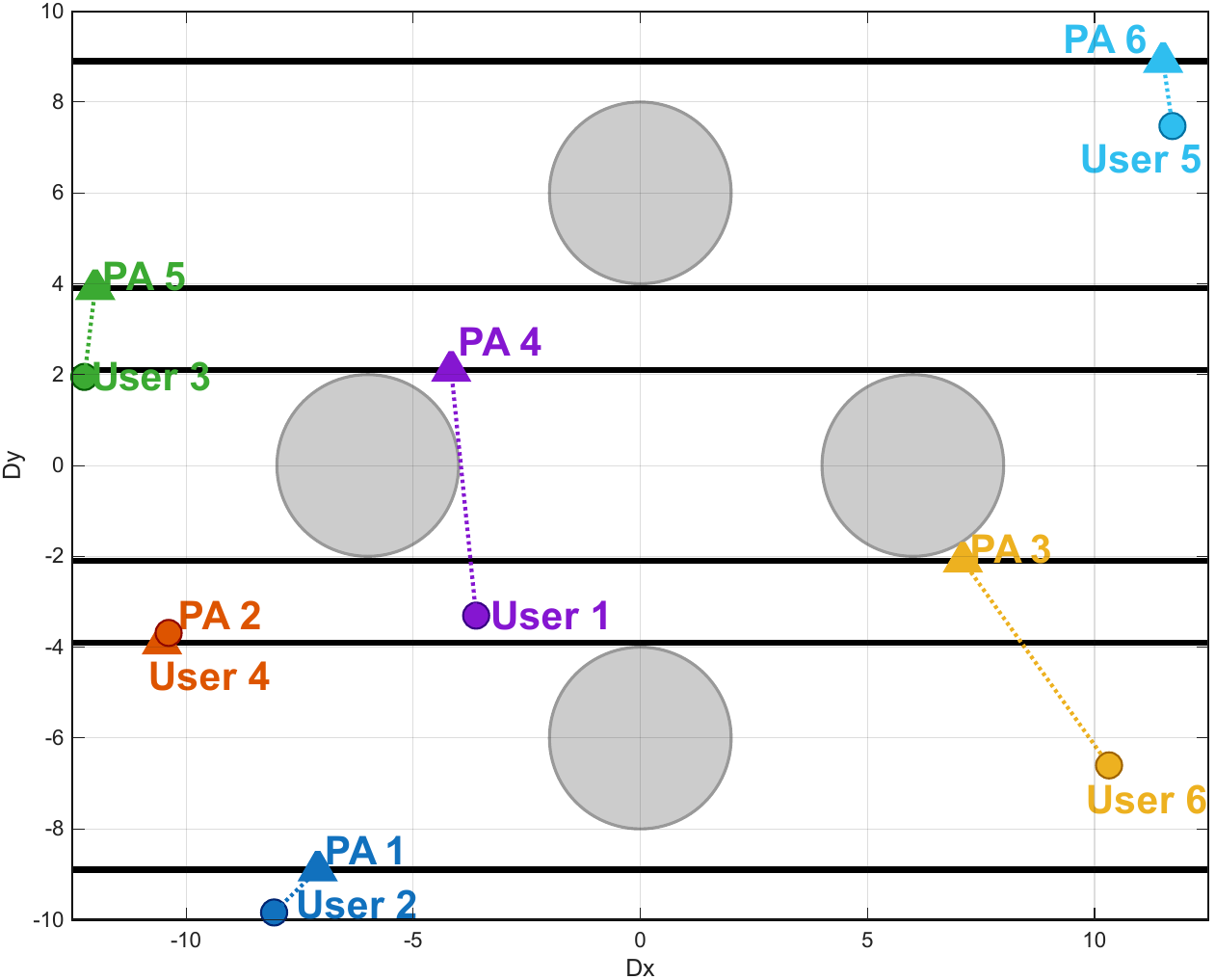}
     \caption{ An example of the proposed solution, where $K=M=6$ and four obstacles with $r=2$ arranged in a diamond-shaped layout.}
     \label{diamond_obs}
\end{figure}

\begin{figure}[t]
     \centering
     \includegraphics[width=0.4\textwidth]{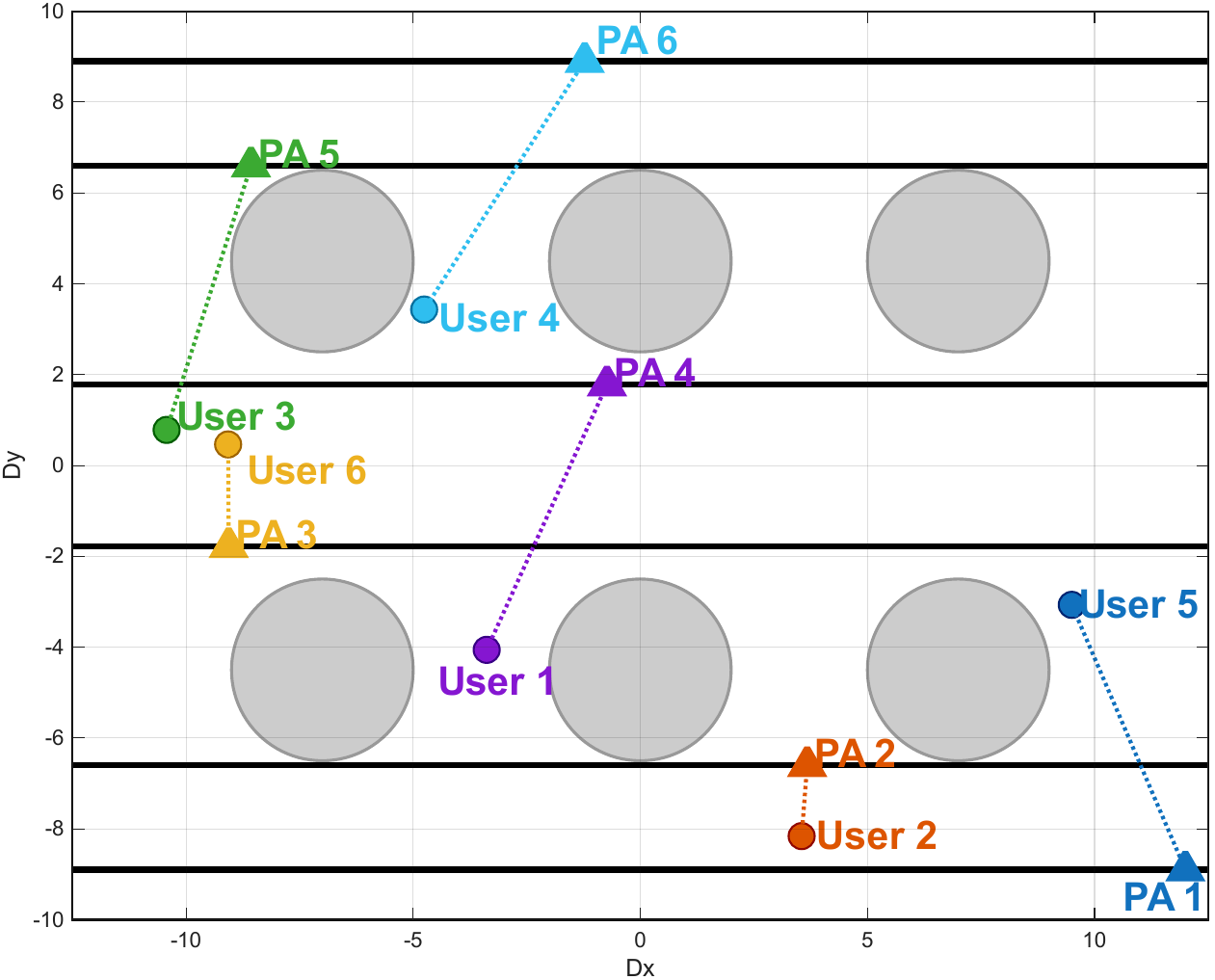}
     \caption{An example of the proposed solution, where $K=M=6$ and six obstacles with $r=2$ arranged in a grid-shaped layout}
     \label{grid_obs}
\end{figure}

Fig. \ref{diamond_obs} and Fig. \ref{grid_obs} illustrate examples of the proposed blockage-aware PA placement solution under two obstacle layouts. In both cases, the system includes $K=M=6$ waveguides and users, with cylinder-shaped obstacles of radius $r=2$~m. Fig. \ref{diamond_obs} shows a diamond-shaped obstacle arrangement. Fig. \ref{grid_obs} presents a grid-shaped obstacle configuration, resulting in denser blockage and more complex propagation paths. In both scenarios, the algorithm adaptively positions PAs to establish LoS links for all users while simultaneously leveraging obstacles to suppress co-channel interference. As a result, it achieves improved sum rate and stable LoS connectivity.

\begin{figure}[t]
     \centering
     \includegraphics[width=0.4\textwidth]{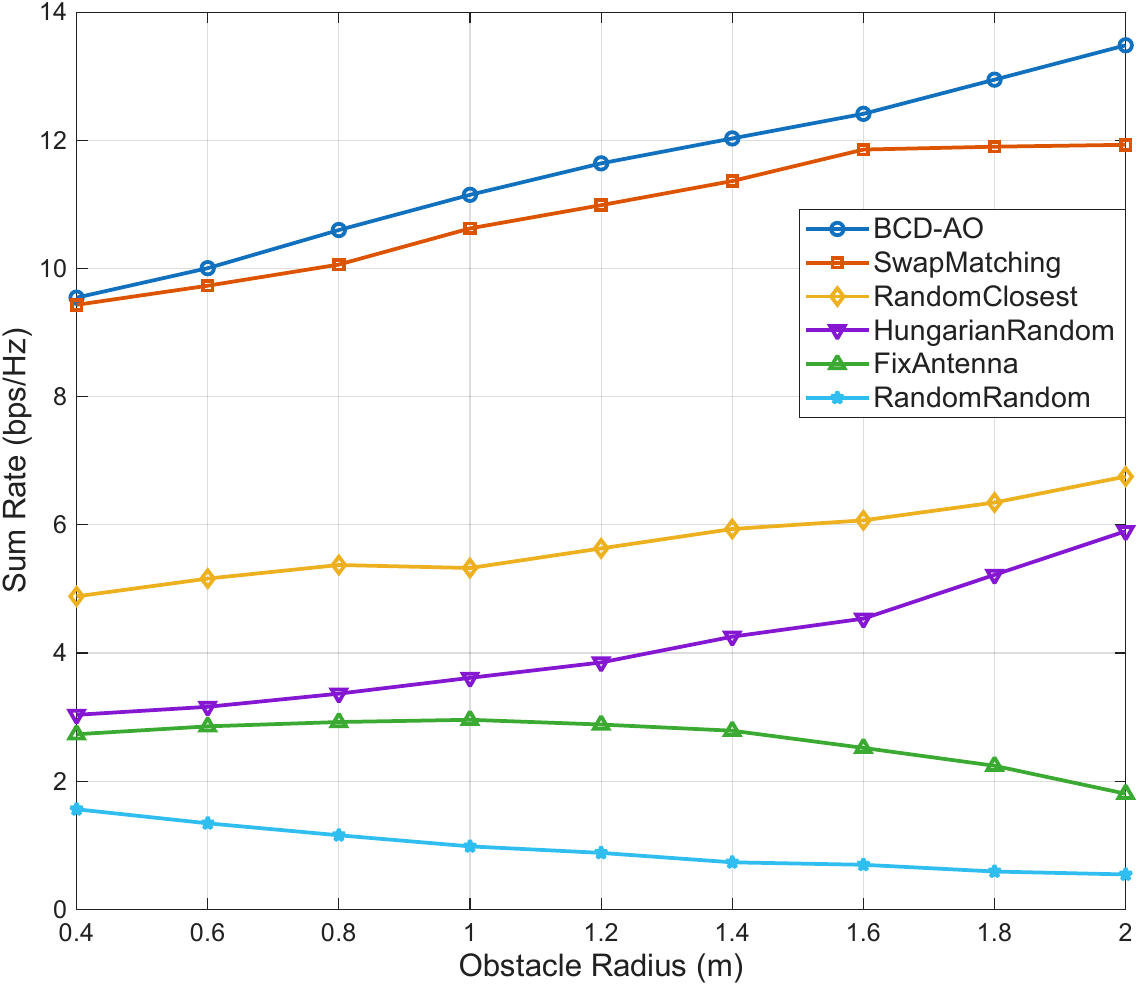}
     \caption{Sum rate performance versus obstacle radius.}
     \label{radius design1}
\end{figure}
Fig. \ref{radius design1} illustrates the sum rate performance versus the obstacle radius. The experimental setting is the same as that in Fig. \ref{grid_obs}. The proposed BCD-AO algorithm achieves the highest throughput across all obstacle sizes, while random and fixed-antenna schemes show much lower performance. The performance gap between adaptive PA-based algorithms and static configurations increases as the obstacle radius grows. An important observation is that the sum rate of the pinching-antenna system improves with larger obstacles, whereas the fixed-antenna system suffers a performance loss. This result highlights a key advantage of pinching antennas that they can dynamically adjust their positions to exploit obstacles for interference suppression. As obstacles become larger, they create stronger blockage and greater spatial diversity, providing more chances for PAs to form LoS connections and avoid co-channel interference, which ultimately enhances system throughput.

\begin{figure}[t]
     \centering
     \includegraphics[width=0.4\textwidth]{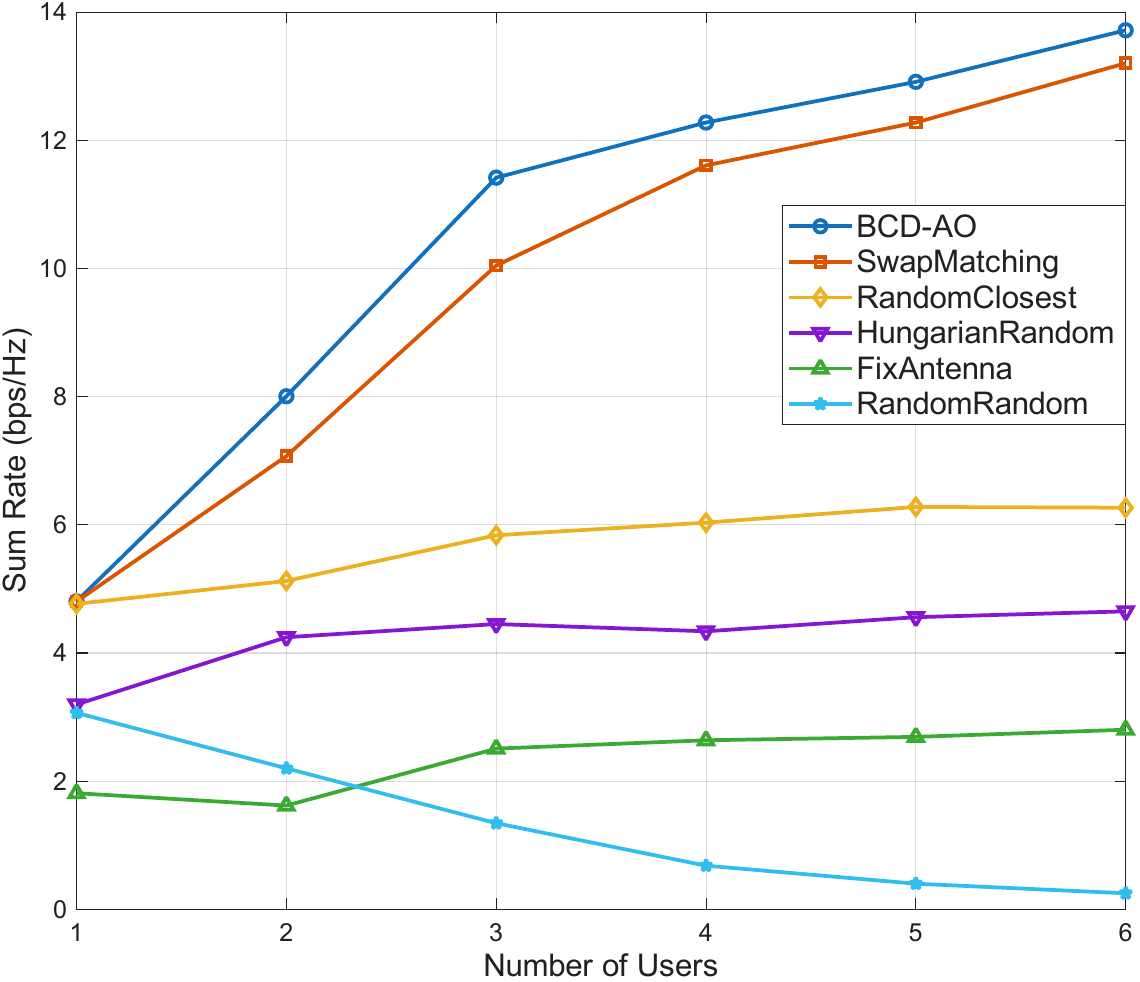}
     \caption{Sum rate performance versus the number of users.}
     \label{user design1}
\end{figure}
Fig. \ref{user design1} illustrates the sum rate performance versus the number of users. The experimental setting is the same as that in Fig. \ref{grid_obs}. The proposed BCD-AO algorithm consistently achieves the highest throughput across all user counts, while the other benchmark schemes show significantly lower performance. As the number of users increases, the sum rate of adaptive PA-based schemes grows rapidly due to their ability to dynamically optimize PA positions and maintain favorable LoS links. In contrast, static configurations such as FixAntenna and RandomRandom exhibit limited scalability because they cannot adapt to increased interference or channel variation. These results indicate that the proposed algorithm can dynamically adapt PA positions to accommodate an increasing number of users without compromising spectral efficiency.

\subsection{Evaluation of the general case}
In this case, each PA can be placed at any position along the waveguide. The proposed algorithm is evaluated against several benchmarks under identical simulation settings. In addition to the proposed WMMSE-based beamforming design, three conventional beamforming schemes are considered for comparison which are maximum ratio combining (MRC), zero-forcing (ZF), and random beamforming. For PA placement optimization, a DDPG-based approach is developed to determine the optimal PA positions. Besides the proposed DDPG method, two baseline schemes are included for performance comparison: a grid-search algorithm and a fixed-antenna system. The grid-search baseline is a coordinate-wise method for optimizing PA positions. It starts from the current PA locations and builds a set of candidate points evenly spaced along the waveguide. A beamforming scheme, such as WMMSE, MRC, ZF, or random beamforming, is first selected to evaluate performance. Then, each PA is optimized one at a time while keeping the others fixed. For each PA, the algorithm tests all candidate positions, computes the corresponding channel and beamforming matrix, and selects the position that gives the highest sum rate. This process is repeated for several passes until no further improvement is observed. The number of candidate positions is set as 25 on each waveguide. In the fixed-antenna system, all antennas are permanently located at the feed point of the waveguide without any movement. \par

Fig. \ref{loss} illustrates the training process of the proposed DDPG agent. The top subfigure shows the reward evolution over training steps, where the reward steadily increases and gradually converges. This indicates that the agent successfully learns an effective policy for PA placement. The bottom subfigure shows the corresponding critic and actor losses. Both losses decrease rapidly during the early training phase and stabilize as learning progresses. Together, these results demonstrate that the DDPG agent efficiently learns the optimal strategy with good convergence properties. \par

\begin{figure}[t]
     \centering
     \includegraphics[width=0.4\textwidth]{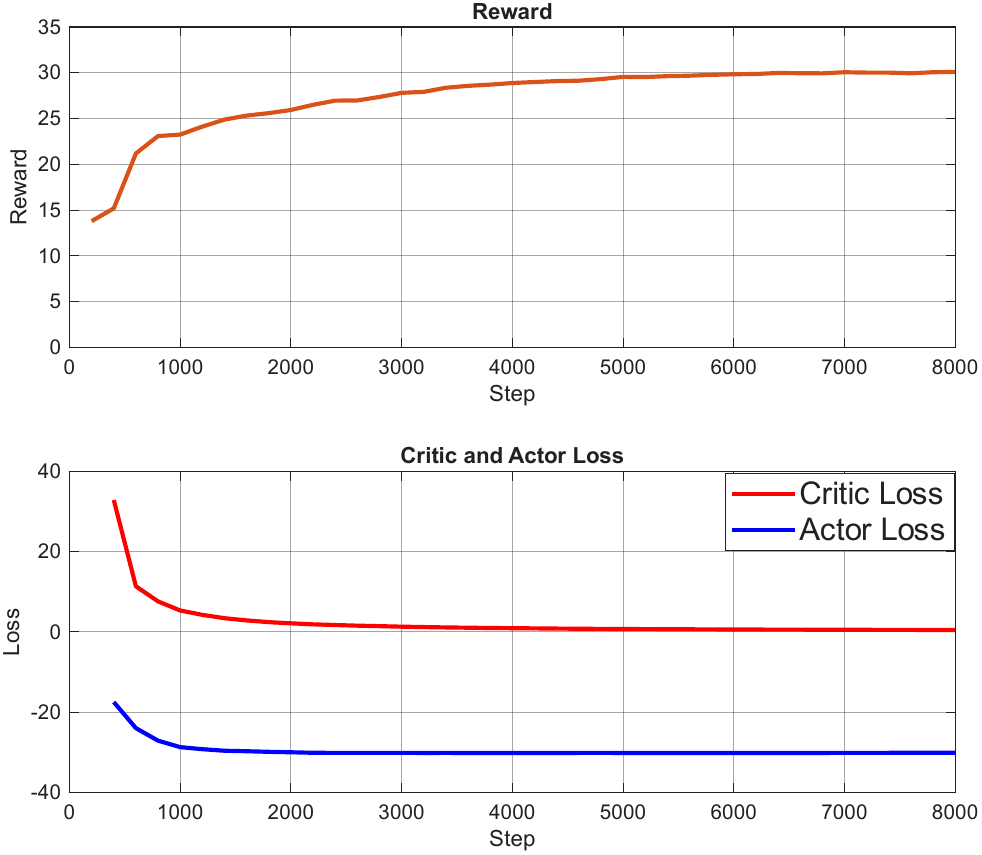}
     \caption{Reward and loss versus training step.}
     \label{loss}
\end{figure}

\begin{figure}[t]
     \centering
     \includegraphics[width=0.4\textwidth]{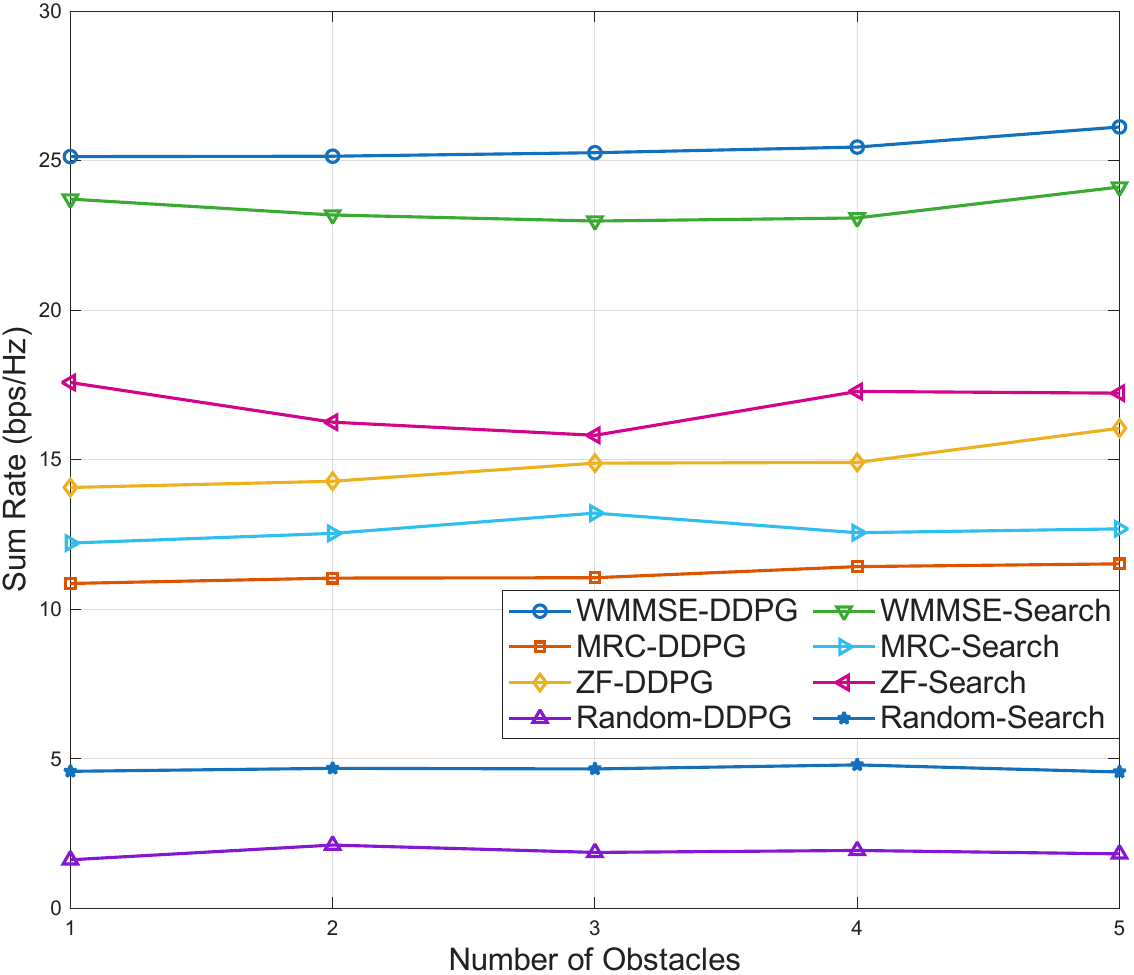}
     \caption{Sum rate performance versus number of obstacles: DDPG and grid-search}
     \label{DDPG and grid-search}
\end{figure}

\begin{figure}[t]
     \centering
     \includegraphics[width=0.4\textwidth]{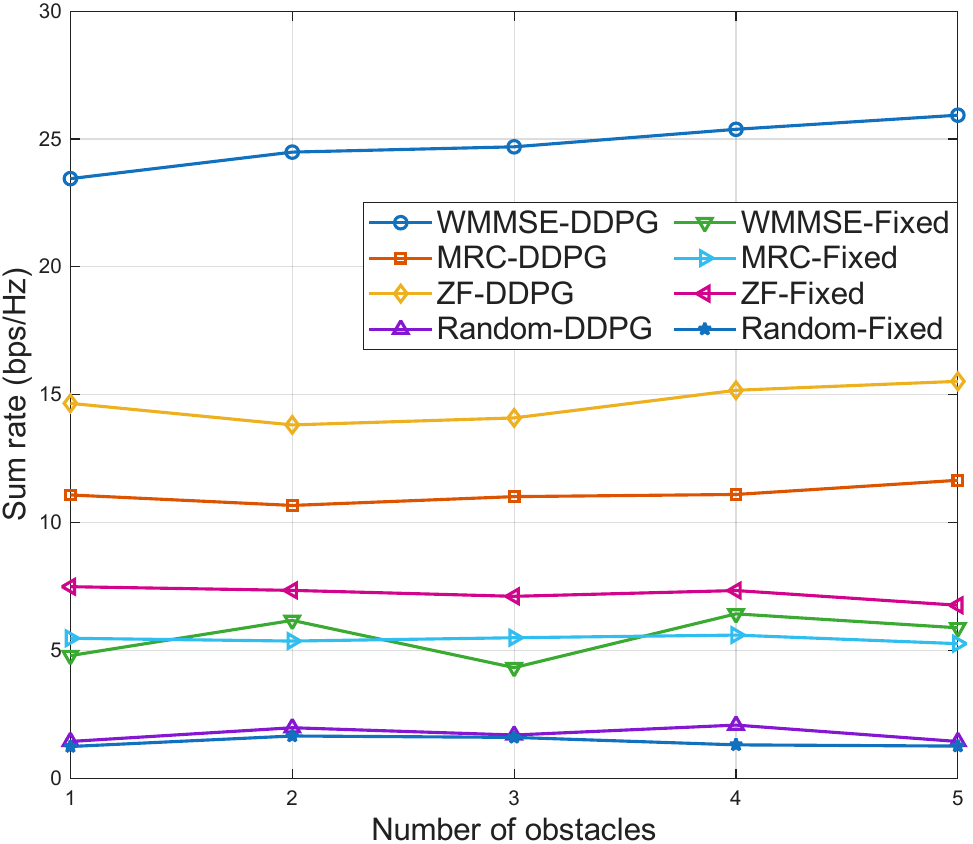}
     \caption{Sum rate performance versus number of obstacles: pinching-antenna systems and fixed-antenna systems.}
     \label{pa and fixed}
\end{figure}

\begin{figure}[t]
     \centering
     \includegraphics[width=0.4\textwidth]{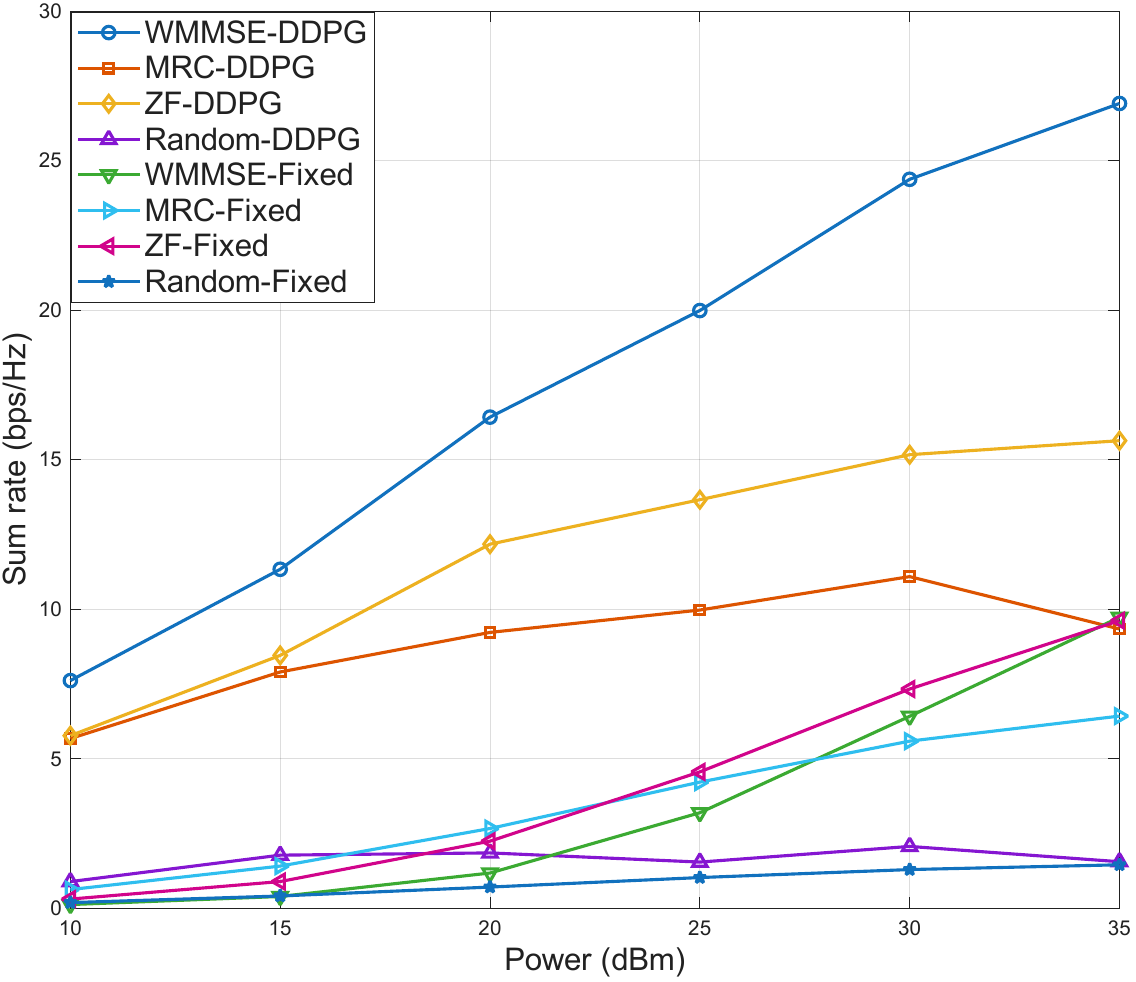}
     \caption{Sum rate performance versus total transmit power}
     \label{power}
\end{figure}

Fig. \ref{DDPG and grid-search} illustrates the sum rate performance versus the number of obstacles for the comparison between learning-based and search-based approaches for optimizing PA positions. The radius of each obstacle is set as 1m and the number of users is set as 4. The proposed WMMSE-DDPG method achieves the highest throughput for all obstacle counts. As the number of obstacles increases, the sum rate of WMMSE-DDPG slightly improves, indicating that the algorithm can leverage obstacles to reduce interference while maintaining LoS links. In contrast, the performance of the search-based and random schemes remains mostly unchanged, which means their limited ability to exploit environmental information. Overall, the results confirm that the proposed learning-based design is more robust to blockage variations and provides better interference management. \par
Fig. \ref{pa and fixed} illustrates the sum rate performance versus the number of obstacles for the comparison between pinching-antenna systems and fixed-antenna system. The experimental setting is the same as that in Fig. \ref{DDPG and grid-search}. As the number of obstacles increases, the sum rate of PA-based systems gradually improves because movable PAs can adjust their positions to maintain LoS links and reduce interference. In contrast, the performance of fixed-antenna systems is much worse than pinching-antenna systems due to increased blockage
and limited adaptability. These results highlight the strong capability of pinching-antenna systems to exploit environmental features for interference mitigation and throughput enhancement. \par
Fig. \ref{power} illustrates the sum rate performance versus total transmit power for the comparison between pinching-antenna systems and fixed-antenna system. The number of obstacles is set as 4 and all obstacles have the same radius 1m. The number of users is set as 4. The proposed WMMSE-DDPG algorithm achieves the highest throughput across all power levels, demonstrating its effectiveness in jointly optimizing beamforming and PA placement. Compared with the baseline schemes, MRC-DDPG and ZF-DDPG achieve moderate gains but suffer from limited interference suppression, especially at higher transmit powers. In contrast, all fixed-antenna schemes show significantly lower performance. As transmit power increases, the performance gap between two antenna systems becomes large, which means that dynamic PA control can better utilize available power to enhance spectral efficiency. \par

\subsection{Discussion}
This work is an initial study on pinching-antenna systems in blockage-aware environments. While the proposed framework establishes a foundation for understanding and optimizing PA behavior under obstruction, several extensions can be explored to make the system more comprehensive and practical. First, the current blockage model focuses on cylinder-shaped obstacles with height not less than $d$. Although this assumption simplifies geometric analysis and visibility determination, future studies can extend the model to more complex obstacle types such as rectangular pillars, furniture, and human bodies. These extensions would allow more accurate characterization of real environments. Second, this work considers only LoS links, where users without LoS connections are assumed to be out of service. A more practical setup could include both LoS and NLoS links, capturing multi-path effects. Third, user locations are assumed to be static during optimization. In realistic scenarios, user mobility continuously changes channel conditions and blockage states. Incorporating mobility-aware control and online PA repositioning would enable dynamic adaptation and enhance system robustness.

\section{Conclusion}
This paper investigated pinching-antenna systems in blockage-aware environments with cylinder-shaped obstacles. A deterministic blockage model was developed to accurately characterize LoS conditions without relying on stochastic assumptions. Based on this model, two PA configurations were examined. In the special case, each PA serves a single user and can only be placed at discrete positions along the waveguide. An alternating optimization framework combining the Hungarian algorithm and a surrogate-assisted BCD search was proposed. In the general case, each PA serves all users with continuously adjustable placement, where a WMMSE-DDPG approach was developed to jointly optimize beamforming and PA positions. This paper demonstrated that pinching-antenna systems can effectively utilize obstacles to mitigate co-channel interference and transform potential blockages into throughput gains, whereas conventional fixed-antenna systems experience performance degradation as blockage increases. This capability makes pinching-antenna systems well suited to obstacle-rich indoor environments. Future work may extend the blockage model to diverse obstacle geometries and incorporate user mobility.

\bibliographystyle{IEEEtran}
\bibliography{IEEEfull,pinchingblockage}
\end{document}